\newif\ifprocs
\procstrue
\procsfalse 

\ifprocs
\documentclass[a4paper,USenglish,cleveref,autoref,thm-restate]{lipics-v2021}
\else
\documentclass{article}
\usepackage[utf8]{inputenc}
\usepackage{fullpage}
\fi

\usepackage{enumitem}
\usepackage{microtype}
\usepackage{lineno}

\usepackage{amssymb}
\usepackage{amsmath}
\usepackage{amsthm}
\usepackage{thmtools}
\usepackage{thm-restate}
\usepackage{mathtools}
\usepackage{csquotes}
\MakeOuterQuote{"}

\usepackage{xcolor}
\ifprocs
\hypersetup{colorlinks=true,allcolors=blue}
\else
\usepackage[colorlinks=true,allcolors=blue]{hyperref}
\usepackage{xspace}
\usepackage{xfrac}
\fi

\usepackage{comment}
\usepackage{algorithm}
\usepackage{algorithmic} 
\usepackage{todonotes}

\ifprocs
\renewcommand{\paragraph}[1]{\subparagraph{#1}}
\newtheorem*{theorem*}{Theorem}
\newtheorem{question}{Question}
\newtheorem{fact}[theorem]{Fact}
\else
\newtheorem{theorem}{Theorem}[section]
\newtheorem*{theorem*}{Theorem}
\newtheorem{conjecture}{Conjecture}

\newtheorem{lemma}[theorem]{Lemma}
\newtheorem{claim}[theorem]{Claim}
\newtheorem{corollary}[theorem]{Corollary}
\newtheorem{proposition}[theorem]{Proposition}

\fi

\newcommand{\ED}{\operatorname{ED}}
\DeclareMathOperator*{\EX}{{\mathbb E}}
\DeclareMathOperator{\poly}{poly}
\DeclareMathOperator{\polylog}{polylog}
\DeclareMathOperator{\pprev}{prev}
\DeclareMathOperator{\nnext}{next}
\DeclareMathOperator{\supp}{supp}
\DeclareMathOperator{\len}{len}
\DeclareMathOperator{\cost}{cost}

\DeclareMathOperator*{\argmin}{arg\,min}
\newcommand\deq{\stackrel{\mathrm{dist}}{=}}

\newcommand{\med}{x_\textrm{med}}
\newcommand{\m}{\textrm{med}}

\newcommand{\N}{\mathbb{N}}

\newcommand{\I}{\mathcal{I}}
\newcommand{\J}{\mathcal{J}}

\newcommand{\mE}{\mathcal{E}}
\newcommand{\tI}{\mathcal{\tilde{I}}}

\newcommand{\obj}{\texttt{Obj}}
\newcommand{\opt}{\texttt{OPT}}

\providecommand{\eqdef}{\coloneqq}
\providecommand{\set}[1]{{\{#1\}}}
\providecommand{\BIGset}[1]{{\Big\{#1\Big\}}}

\def\compactify{\itemsep=0pt \topsep=0pt \partopsep=0pt \parsep=0pt}

\newcommand{\colnote}[3]{\textcolor{#1}{$\ll$\textsf{#2}$\gg$\marginpar{\tiny\bf #3}}}
\newcommand{\rnote}[1]{\colnote{red}{#1--Robi}{RK}}

\title{Approximate Trace Reconstruction via Median String (in Average-Case)}

\ifprocs
\author{Diptarka Chakraborty}{National University of Singapore }{diptarka@comp.nus.edu.sg}{}{[Work partially supported by NUS ODPRT Grant, WBS No. R-252-000-A94-133.]}

\author{Debarati Das}{Basic Algorithm Research Copenhagen (BARC), University of Copenhagen}{debaratix710@gmail.com}{}{}

\author{Robert Krauthgamer}{Weizmann Institute of Science}{robert.krauthgamer@weizmann.ac.il}{}{[Work partially supported by ONR Award N00014-18-1-2364, the Israel Science Foundation grant \#1086/18, and a Minerva Foundation grant,
  and by the Israeli Council for Higher Education (CHE) via the Weizmann Data Science Research Center.]}

\authorrunning{D.~Chakraborty and D.~Das and R.~Krauthgamer} 

\else
\author{%
Diptarka Chakraborty%
\thanks{National University of Singapore.
  Work partially supported by NUS ODPRT Grant, WBS No. R-252-000-A94-133.
    Email: \texttt{diptarka@comp.nus.edu.sg}
  }
\and
Debarati Das%
\thanks{Basic Algorithm Research Copenhagen (BARC), University of Copenhagen
        Email: \texttt{debaratix710@gmail.com}
}
\and
Robert Krauthgamer%
  \thanks{Weizmann Institute of Science.
    Work partially supported by ONR Award N00014-18-1-2364, the Israel Science Foundation grant \#1086/18, and a Minerva Foundation grant,
    and by the Israeli Council for Higher Education (CHE) via the Weizmann Data Science Research Center. 
    Email: \texttt{robert.krauthgamer@weizmann.ac.il}
  }
}
\fi

\ifprocs
\Copyright{Debarati Das and Diptarka Chakraborty and Robert Krauthgame} 



\EventEditors{John Q. Open and Joan R. Access}
\EventNoEds{2}
\EventLongTitle{42nd Conference on Very Important Topics (CVIT 2016)}
\EventShortTitle{CVIT 2016}
\EventAcronym{CVIT}
\EventYear{2016}
\EventDate{December 24--27, 2016}
\EventLocation{Little Whinging, United Kingdom}
\EventLogo{}
\SeriesVolume{42}
\ArticleNo{23}
\else
\fi

\begin{document}

\maketitle

\thispagestyle{empty}
\setcounter{page}{0}

\begin{abstract}
We consider an \emph{approximate} version of the trace reconstruction problem,
where the goal is to recover an unknown string $s\in\{0,1\}^n$ 
from $m$ traces 
(each trace is generated independently by passing $s$ through a probabilistic 
insertion-deletion channel with rate $p$). 
We present a deterministic near-linear time algorithm for the average-case model, 
where $s$ is random, that uses only \emph{three} traces.
It runs in near-linear time $\tilde O(n)$
and with high probability reports a string 
within edit distance $O(\epsilon p n)$ from $s$ for $\epsilon=\tilde O(p)$,
which significantly improves over the straightforward bound of $O(pn)$.

Technically, our algorithm computes a $(1+\epsilon)$-approximate median
of the three input traces. 
To prove its correctness, our probabilistic analysis shows 
that an approximate median is indeed close to the unknown $s$.
To achieve a near-linear time bound, we have to bypass the well-known 
dynamic programming algorithm that computes an optimal median in time $O(n^3)$.
\end{abstract}

\newpage

\section{Introduction}
\label{sec:intro}

\paragraph{Trace Reconstruction.}
One of the most common problems in statistics is to estimate an unknown parameter from a set of noisy observations (or samples). The main objectives are (1) to use as few samples as possible, (2) to minimize the estimation error, and (3) to design an efficient estimation algorithm. 
One such parameter-estimation problem is \emph{trace reconstruction}, 
where the unknown quantity is a string $s \in \Sigma^n$,
and the observations are independent \emph{traces}, 
where a trace is a string that results from $s$ passing through some noise channel. The goal is to reconstruct $s$ using a few traces. (Unless otherwise specified, in this paper we consider $\Sigma=\{0,1\}$.) Various noise channels have been considered so far. 
The most basic one only performs substitutions. 
A more challenging channel performs deletions. Even more challenging is the \emph{insertion-deletion} channel, 
which scans the string $s$ 
and keeps the next character with probability $1-p$, 
deletes it with probability $p/2$, 
or inserts a uniformly randomly chosen symbol (without processing the next character) with probability $p/2$,
for some noise-rate parameter $p\in [0,1)$. 
We denote this insertion-deletion channel by $R_p(s)$,
see Section~\ref{sec:prob-model} for a formal definition.%
\footnote{In the literature, an insertion-deletion channel with different probabilities for insertion and for deletion has been studied. For simplicity in exposition, we consider a single error probability throughout this paper, however our results can easily be generalized to different insertion and deletion probabilities. Another possible generalization is to allow substitutions along with insertions and deletions. Again, for simplicity, we do not consider substitutions, but with slightly more careful analysis our results could be extended.
}

The literature studies mostly two variants of trace reconstruction.
In the \emph{worst-case} variant, the unknown string $s$ is an arbitrary string from $\Sigma^n$,
while in the \emph{average-case} variant, $s$ is assumed to be drawn uniformly at random from $\Sigma^n$. 
The trace reconstruction problem finds numerous applications in computational biology, DNA storage systems, coding theory, etc. Starting from early 1970s~\cite{Kal73}, various other versions of this problem have been studied, 
including combinatorial channels~\cite{levenshtein2001efficient, levenshtein2001efficient2}, 
smoothed complexity~\cite{CDLSS21}, 
coded trace reconstruction~\cite{cheraghchi2020coded}, 
and population recovery~\cite{ban2019beyond, narayanan2020population}. 

We focus on the average-case variant, 
where it is known that $\exp(O(\log^{1/3} n))$ samples suffice to reconstruct a (random) unknown string $s$ over the insertion-deletion channel~\cite{HPP20}.
On the other hand, a recent result~\cite{C20b} showed that $\tilde{\Omega}(\log^{5/2}n)$ samples are necessary, improving upon the previous best lower bound of $\tilde{\Omega}(\log^{9/4}n)$~\cite{HL20}. 
We emphasize that all these upper and lower bounds are for \emph{exact} trace reconstruction, i.e., for recovering the unknown string $x$ perfectly (with no errors). 
A natural question proposed by Mitzenmacher~\cite{mitzenmacher2009survey} is whether such a lower bound on the sample complexity can be bypassed by allowing approximation,
i.e., by finding a string $z$ that is "close" to the unknown string $s$. 
One of the most fundamental measures of closeness between a pair of strings $z$ and $z'$, is their \emph{edit distance}, denoted by $\ED(z,z')$ and defined as the minimum number of insertion, deletion, and substitution operations needed to transform $z$ into $z'$. 
Observe that a trace generated from $s$ via an insertion-deletion channel $G_p$ has expected edit distance about $pn$ from the unknown string $s$
(see Section~\ref{sec:unique-edit}). We ask how many traces (or samples) are required to construct a string $z$ at a much smaller edit distance from the unknown $s$. 
(Since the insertion-deletion channel has no substitutions, we also do not consider substitutions in our analysis of the edit distance, however the results probably hold also when both allow also substitutions.)

A practical application of average-case trace reconstruction is in the portable DNA-based data storage system. In the DNA storage system~\cite{GBCDLSB13, RMRAJY17}, a file is preprocessed by encoding it into a DNA sequence. This encoded sequence is randomized using a pseudo-random sequence, and thus the final encoding sequence could be treated as a (pseudo-)random string. The stored (encoded) data is retrieved using next-generation sequencing (like single-molecule real-time sequencing (SMRT)~\cite{RCS13} that involves $12-18\%$, which generates several noisy copies (traces) of the stored data via some insertion-deletion channel. The final step is to decode back the stored data with as few traces as possible. Currently, researchers use multiple sequence alignment algorithms to reconstruct the trace~\cite{yazdi2017portable, organick2018random}. Unfortunately, such heuristic algorithms are notoriously difficult to analyze rigorously to show a theoretical guarantee. However, the preprocessing step also involves error-correcting code to encode the strings. Thus it suffices to reconstruct the original string up to some small error (depending on the error-correcting codes used). This specific application gives one motivation to study approximate trace reconstruction.

Our main contribution is to show that it is sufficient to use only three traces to reconstruct the unknown string up to a small edit error. 
\begin{theorem}
\label{thm:main}
There is a constant $c_0>0$ and a deterministic algorithm that, given as input
a noise parameter $p \in (0,c_0]$,
and three traces from the insertion-deletion channel $R_p(s)$
for a uniformly random (but unknown) string $s\in \{0,1\}^n$, 
outputs in time $\tilde{O}(n)$ a string $z$ that satisfies 
$\Pr [\ED(s,z) \le O(p^2 \log(1/p) n) ] \geq 1-n^{-1}$.
\end{theorem}
The probability in this theorem is over the random choice of $s$ and the randomness of the insertion-deletion channel $R_p$. We note that the term $\log(1/p)$ in the estimation error $\ED(s,z)$ can be shaved by increasing the alphabet size to $\poly(1/\epsilon)$. 
An edit error of $O(p^2n)$ is optimal for three traces,
because in expectation $O(p^2 n)$ characters of $s$ are deleted in two of the three traces, and look as if they are inserted to $s$ in one of the three traces 
(which occurs in expectation for even more characters).

Our theorem demonstrates that the number of required traces 
exhibits a sharp contrast between exact and approximate trace reconstruction. 
In fact, approximate reconstruction not only beats 
the $\Omega(\log^{5/2} n)$ lower bound for exact reconstruction, 
but surprisingly uses only \emph{three} traces! 
We conjecture that the estimation error $\ED(s,z)$ 
can be reduced further using more than three traces. 
We believe that our technique can be useful here,
but this is left open for future work. 
\begin{conjecture}
\label{conj:main}
The estimation error $\ED(s,z)$ in Theorem~\ref{thm:main} 
can be reduced to $O(\epsilon pn)$ for arbitrarily small $\epsilon > 0$, 
using $\poly(1/\epsilon)$ traces. 
\end{conjecture}
This conjecture holds for $\epsilon<1/n$,
as follows from known bounds for exact reconstruction~\cite{HPP20},
and perhaps suggests that a number of traces that is sub-polynomial in $1/\epsilon$
it suffices for all $\epsilon$. 

\paragraph{Median String.}
As mentioned earlier,
a common heuristic to solve the trace reconstruction problem is multiple sequence alignment,
which can be formulated equivalently (see~\cite{gusfield1997} and the references therein) as the problem of finding a median under edit distance. 
For general context, the median problem is a classical aggregation task in data analysis;
its input is a set $S$ of points in a metric space relevant to the intended application, 
and the goal is to find a point (not necessarily from $S$) 
with the minimum sum of distances to points in $S$, i.e., 
\begin{equation} \label{eq:median}
  \min_y \sum_{x\in S} d(y,x).    
\end{equation}
Such a point is called a \emph{median} (or \emph{geometric median} in a Euclidean space). 
For many applications, it suffices to find an \emph{approximate median}, 
i.e., a point in the metric with approximately minimal objective value~\eqref{eq:median} . 
The problem of finding an (approximate) median has been studied extensively both in theory and in applied domains, over various metric spaces, including Euclidean~\cite{cohen2016geometric} (see references therein for an overview), Hamming (folklore), the edit metric~\cite{Sankoff75, kruskal1983, NR03}, rankings~\cite{DKNS01,ACN08,MS07}, Jaccard distance~\cite{CKPV10}, Ulam~\cite{chakraborty2021approximating}, 
and many more~\cite{fletcher2008robust, minsker2015geometric, cardot2017online}.

The median problem over the \emph{edit-distance metric} is known as the \emph{median string} problem~\cite{kohonen1985median},
and finds numerous applications in computational biology~\cite{gusfield1997, pevzner2000computational}, DNA storage system~\cite{GBCDLSB13, RMRAJY17}, speech recognition~\cite{kohonen1985median}, and classification~\cite{martinez2000use}. 
This problem is known to be NP-hard~\cite{HC00, NR03} (even W[1]-hard~\cite{NR03}), 
and can be solved by standard dynamic programming~\cite{Sankoff75, kruskal1983}
in time $O(2^m n^m)$ when the input has $m=|S|$ strings of length $n$ each. 
From the perspective of approximation algorithms, 
a multiplicative $2$-approximation to the median is straightforward 
(this works in every metric space by simply reporting the best among the input strings, i.e., $y^*\in S$ that minimizes the objective). 
However, no polynomial-time algorithm is known to break below $2$-approximation (i.e., achieve factor $2-\delta$ for fixed $\delta > 0$) for the median string problem,
despite several heuristic algorithms and results for special cases~\cite{casacuberta1997greedy, kruzslicz1999improved, fischer2000string, pedreira2007spatial, abreu2014new, hayashida2016integer, Mirabal19, chakraborty2021approximating}.

Although the median string (or equivalently multiple sequence alignment)
is a common heuristic for trace reconstruction~\cite{yazdi2017portable, organick2018random},
to the best of our knowledge there is no definite connection between these two problems. 
We show that both the problems are roughly the same in the average-case model. It is not difficult to show that any string close to the unknown string is an approximate median. To see this, we can show that for a set $S$ of $m$ traces of an (unknown) random string $s$, their optimal median objective value is at least $(1-O(\epsilon))pnm$ with high probability. On the other hand, the median objective value with respect to $s$ itself is at most $(1+\epsilon)pnm$ with high probability. (See the proof of Claim~\ref{clm:opt-value}.) 
Hence, the unknown string $s$ is an $(1+O(\epsilon))$-approximate median of $S$,
and by the triangle inequality, every string close (in edit distance) to $s$ is also an approximate median of $S$. 
One of the major contributions of this paper is the converse direction,
showing that given a set of traces of an unknown string, 
any approximate median of the traces is close (in edit distance) to the unknown string.
This is true even for three traces.
\begin{theorem}
\label{thm:median-trace-close}
For a large enough $n \in \N$ and a noise parameter $p \in (0,0.001)$, 
let the string $s\in\{0,1\}^n$ be chosen uniformly at random, 
and let $s_1,s_2,s_3$ be three traces generated by the insertion-deletion channel $R_p(s)$. 
If $\med$ is a $(1+\epsilon)$-approximate median of $\{s_1,s_2,s_3\}$
for $\epsilon \in [110 p \log (1/p),1/6]$, 
then $\Pr [\ED(s,\med) \le O(\epsilon)\cdot \opt ] \geq 1-n^{-3}$,  
where $\opt$ denotes the optimal median objective value of $\{s_1,s_2,s_3\}$. 
\end{theorem}
An immediate consequence (see Corollary~\ref{cor:m-traces-median}) is that for every $3\le m < n^{O(1)}$ traces, every $(1+\epsilon)$-approximate median $\med$ satisfies $\ED(s,\med) \le O(\epsilon) \frac{\opt}{m}$.

Thus if we could solve any of the two problems (even approximately), we get an approximate solution to the other problem. E.g., the current best (exact) trace reconstruction algorithm for the average-case~\cite{HPP20} immediately gives us an $O(n^{1+o(1)})$ time algorithm to find an $(1+O(\epsilon))$-approximate median of a set of traces. 
We leverage this interplay between the two problems to design an efficient algorithm for approximate trace reconstruction. 
Since one can compute the (exact) median of three strings $s_1,s_2,s_3$ in time $O(|s_1|\cdot|s_2|\cdot|s_3|)$~\cite{Sankoff75, kruskal1983}, 
the above theorem immediately provides us the unknown string up to some small edit error in time $O(n^3)$. We further reduce the running time to near-linear by cleverly partitioning each of the traces into $\polylog n$-size blocks and then applying the median algorithm on these blocks. Finally, we concatenate all the block-medians to get an "approximate" unknown string, leading to Theorem~\ref{thm:main}. One may further note that Theorem~\ref{thm:main} also provides a $(1+O(\epsilon))$-approximate median for any set of traces in the average-case (again due to Theorem~\ref{thm:median-trace-close}).

Taking the smallest possible $\epsilon$ in Theorem~\ref{thm:median-trace-close}, 
we get that for three traces generated from $s$, with high probability $\ED(s,\med) \le \tilde{O}(p^2 n)$.
In comparison, it is not hard to see that with high probability $\opt$ is bounded by roughly $3pn$.
We conjecture that the number of traces increases,
the median string converges to the unknown string $s$. 
In particular, $\ED(s,\med)\leq \epsilon n$ when using $\poly(1/\epsilon)$ traces (instead of just three), with high probability. 
We hope that our technique can be extended to prove the above conjecture, but we leave it open for future work. 

The main implication of this conjecture is an $\tilde{O}(n)$ time approximate trace reconstruction algorithm, for any fixed $\epsilon>0$, as follows. 
It is straightforward to extend our approximate median finding algorithm (in Section~\ref{sec:trace-median-algo}) to more input strings. 
(For brevity, we present only for three input strings.) 
For $m$ strings, the running time would be $n (\log n)^{O(m)}$, 
and for $m=\poly(1/\epsilon)$ strings this running time is $n \polylog n$. 
As a consequence, we will be able to reconstruct in $\tilde{O}(n)$ time
a string $z$ such that $\ED(s,z) \le \epsilon n$, 
which in particular implies Conjecture~\ref{conj:main}.

\subsection{Related Work}
A systematic study on the trace reconstruction problem has been started since~\cite{levenshtein2001efficient, levenshtein2001efficient2, BKKM04}. However, some of its variants appeared even in the early '70s~\cite{Kal73}. One of the main objectives here is to reduce the number of traces required, aka the sample complexity. Both the deletion only and the insertion-deletion channels have been considered so far. In the general worst-case version, the problem considers the unknown string $s$ to be any arbitrary string from $\{0,1\}^n$. The very first result by Batu et al.~\cite{BKKM04} asserts that for small deletion probability (noise parameter) $p\le \frac{1}{n^{1/2+\epsilon}}$, to reconstruct $s$ considering $O(n\log n)$ samples suffice. A very recent work~\cite{CDLSS21b} improved the sample complexity to $\mathrm{poly}(n)$ while allowing a deletion probability $p\le \frac{1}{n^{1/3+\epsilon}}$.
For any constant deletion probability bounded away from 1, the first subexponential (more specifically, $2^{\tilde{O}(\sqrt{n})}$) sample complexity was shown by~\cite{HMPW08}, which was later improved to $2^{O(n^{1/3})}$~\cite{NP17, de2017optimal}, and then finally to $2^{O(n^{1/5})}$~\cite{C20}.

Another natural variant that has also been widely studied is the average-case, where the unknown string $s$ is randomly chosen from $\{0,1\}^n$. It turns out that this version is significantly simpler than the worst-case in terms of the sample complexity. For sufficiently small noise parameter ($p=o(1)$ as a function of $n$), efficient trace reconstruction algorithms are known~\cite{BKKM04, kannan2005more, viswanathan2008improved}. For any constant noise parameter bounded away from 1 in case of insertion-deletion channel, the current best sample complexity is $\exp(O(\log^{1/3}n))$~\cite{HPP20} improving up on $\exp(O(\log^{1/2}n))$~\cite{PZ17}. Both of these results are built on the worst-case trace reconstruction by~\cite{NP17, de2017optimal}. Furthermore, the trace reconstruction algorithm of~\cite{HPP20} runs in $n^{1+o(1)}$ time.

In the case of the lower bound, information-theoretically, it is easy to see that $\Omega(\log n)$ samples must be needed when the deletion probability is at least some constant. In the worst-case model, the best known lower bound on the sample complexity is $\tilde{\Omega}(n^{3/2})$~\cite{C20b}. For the average-case, McGregor, Price, and Vorotnikova~\cite{MPV14} showed that $\Omega(\log^2 n)$ samples are necessary to reconstruct the unknown (random) string $s$. This bound was further improved to $\tilde{\Omega}(\log^{9/4}n)$ by Holden and Lyons~\cite{HL20}, and very recently to $\tilde{\Omega}(\log^{5/2}n)$ by Chase~\cite{C20b}.

The results described above show an exponential gap between the upper bound and lower bound of the sample complexity. The natural question is, instead of reconstructing the unknown string exactly, if we allow some error in the reconstructed string, then can we reduce the sample complexity? Recently, Davies et al.~\cite{DRRS20} presented an algorithm that for a specific class of strings (considering various run-lengths or density assumptions), can compute an approximate trace with $\epsilon n$ additive error under the edit distance while using only $\mathrm{polylog}(n)$ samples. The authors also established that to approximate within the edit distance $n^{1/3-\delta}$, the number of required samples is $n^{1 + 3\delta/2}/\mathrm{polylog}(n)$, for $0<\delta<1/3$, in the worst case. Independently, Grigorescu et al.~\cite{GSZ20} showed assuming deletion probability $p=1/2$, there exist two strings within edit distance 4 such that any \emph{mean-based algorithm} requires $\exp(\Omega(\log^2 n))$ samples to distinguish them.

\subsection{Technical Overview}
The key contribution of this paper is a linear-time approximate trace reconstruction algorithm that uses only three traces to reconstruct an unknown (random) string up to some small edit error (Theorem~\ref{thm:main}). To get our result, we establish a relation between the (approximate) trace reconstruction problem and the (approximate) median string problem. 
Consider a uniformly random (unknown) string $s \in \Sigma^n$. 
We show that for any three traces of $s$ generated by the probabilistic insertion-deletion channel $R_p$, 
an arbitrary $(1+\epsilon)$-approximate median of the three trace must be, 
with high probability, 
$O(\epsilon)\opt$-close in edit distance to $s$  (Theorem~\ref{thm:median-trace-close}). 
Once we establish this connection, 
it suffices to solve the median problem (even approximately).
The median of three traces can be solved optimally in $O(n^3)$ time using 
a standard dynamic programming algorithm~\cite{Sankoff75, kruskal1983}. 
It is not difficult to show that the optimal median objective value $\opt$ 
is at least $3(1 -O(\epsilon))p n$ (Claim~\ref{clm:opt-value}),
and thus the computed median is at edit distance at most $O(\epsilon p n)$ 
from the unknown string $s$. 
This result already beats the known lower bound for \emph{exact} trace reconstruction in terms of sample complexity. 
However, the running time is cubic in $n$, whereas current average-case trace reconstruction algorithms run in time $n^{1+o(1)}$~\cite{HPP20}.

Next we briefly describe the algorithm that improves the running time to $\tilde{O}(n)$. 
Instead of finding a median of the entire traces, 
we compute the median block-by-block and then concatenate the resulting blocks. 
A natural idea is that each such block is just the median of three substrings  
taken from the three traces, 
but the challenge is to identify which substring to take from each trace,
particularly because $s$ is not known. 
To mitigate this issue, we take the first trace $s_1$ and partition it into disjoint blocks of length $\Theta(\log^2 n)$ each. 
For each such block, we consider its middle $\log^2 n$-size sub-block as an \emph{anchor}.
We then locate for each anchor its corresponding substrings in the other two traces $s_2$ and $s_3$, 
using any approximate pattern matching algorithm under the edit metric (e.g.~\cite{LV89, GP90}) to find the \emph{best match} of the anchor inside $s_2,s_3$.
Each anchor has a \emph{true match} in $s_2$ and in $s_3$, 
i.e., the portion that the anchor generated under the noise channel. 
Since the anchors in $s_1$ are "well-separated" (by at least $\omega(\log n)$), 
their true matches are also far apart both in $s_2,s_3$. 
Further, exploiting the fact that $s$ is a random string, we can argue that each anchor's best match and true match overlap almost completely 
(i.e., except for a small portion), see Section~\ref{sec:trace-median-algo}). 
We thus treat these best match blocks as anchors in $s_2$ and $s_3$ and partition them into blocks. 
From this point, the algorithm is straightforward. Just consider the first block of each of $s_1,s_2,s_3$ and compute their median. Then consider the second block from each trace and compute their median, and so on. Finally, concatenate all these block medians, and output the resulting string. 

The crucial claim is that the best match and true match of an anchor in $s_1$ 
are the same except for a small portion, is crucial from two aspects. 
First, it ensures that any $r$-th block of $s_2,s_3$ contains the true match of the $r$-th anchor of $s_1$. 
Consequently, computing a median of these blocks reconstructs the corresponding portion of the unknown string $s$ up to edit distance $O(\epsilon) p \log^2 n$ with high probability. 
Thus for "most of the blocks", we can reconstruct up to such edit distance bound. 
We can make the length of the non-anchor portions negligible compared to the anchors (simply because a relatively small "buffer" around each anchor suffices), 
and thus we may ignore them and still ensure that the output string is $O(\epsilon p n)$-close (in edit distance) to the unknown string $s$. 
(See the proof of Lemma~\ref{lem:correctness} for the details.) 
The second use of that crucial claim is that it helps in searching for the best match of each anchor "locally" (within a $O(\log^2 n)$-size window) in each $s_j$, $j \in \{2,3\}$. As a result, we bound the running time of the pattern matching step by $\tilde{O}(n)$. The median computations are also on $\Theta(\log^2 n)$-size blocks, and thus takes a total $\tilde{O}(n)$ time.

It remains to explain the key contribution, which is the connection between the (approximate) trace reconstruction and the (approximate) median string problem. 
Its first ingredient is that there is an "almost unique" alignment between the unknown string $s$ and a trace of it generated by the insertion-deletion channel $R_p$. 
We provide below an overview of this analysis 
(see Section~\ref{sec:unique-edit} for details).
We start by considering the random string $s$ and a string $y$ generated by passing $s$ through the noise channel $R_p$. 
For sake of analysis, 
we can replace $R_p$ with an \emph{equivalent} probabilistic model $G_p$, 
that first computes a random alignment $A^p$ between $s$ and $y$,
and only then fills in random characters in $s$ and in the insertion-positions in $y$.
This model is more convenient because it separates the two sources of randomness,
for example we can condition on one ($A^p$) 
when analyzing typical behavior of the other (characters of $s$).

\begin{figure}[tp]
    \centering
    \includegraphics[scale=0.45]{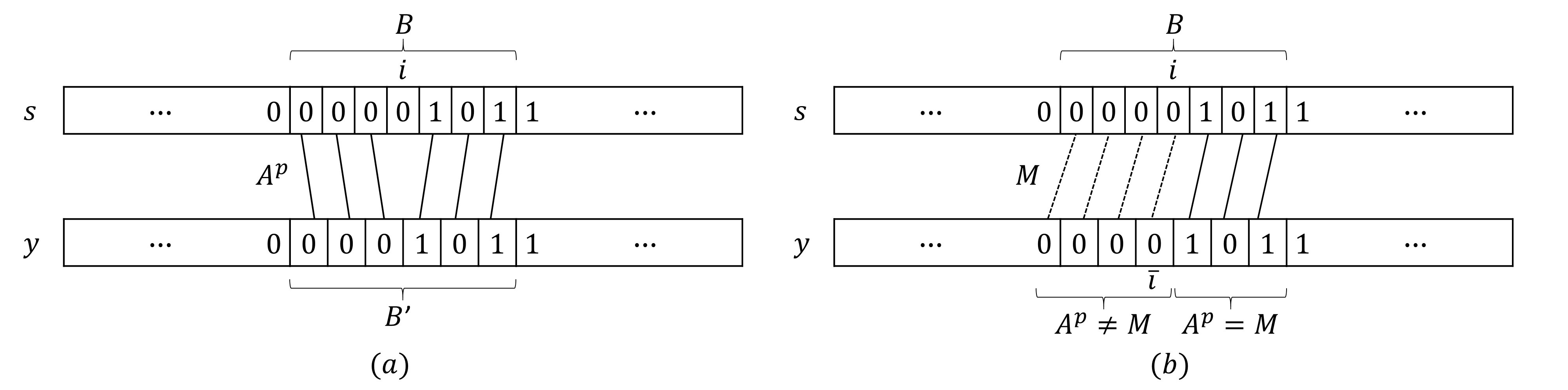}
    \caption{(a) An example of well separated edit operation: $A^p$ deletes $s[i]$ and aligns rest of the characters in block $B$ with block $B'$. (b) $M$ aligns $s[i]$ and $y[\bar{i}]$. For each index $j$ appearing left of $i$ in $B$, $A^p[j]\neq M(j)$. }
    \label{fig:fig1}
\end{figure}

In expectation, the channel $G_p$ generates a trace $y$ by performing 
about $pn$ random edit operations in $s$ (planting insertions/deletions),
hence the planted alignment $A^p$ has expected cost about $pn$. 
But can these edit operations cancel each other?
Can they otherwise interact, leading to the optimal edit distance being smaller? 
For example, suppose $s[i]=0$. 
If $G_p$ first inserts a $0$ before $s[i]$ and then deletes $s[i]$, 
then clearly these two operations cancel each other. 
We show that such events are unlikely. 
Following this intuition, we establish our first claim, that with high probability 
and the edit distance between $s,y$ is large,
specifically $\ED(s,y)\ge (1-6\epsilon)pn$ for $\epsilon\ge 15p\log (1/p)$, 
see Lemma~\ref{lem:editbound-small-alphabet};
thus, the planted alignment $A^p$ is near-optimal. 
Towards proving this, 
we first show that a vast majority of the planted edit operations are well-separated,
i.e., have $\Theta(1/p)$ positions between them. 
In this case, for one operation to cancel another one, the characters appearing between them in $s$ must all be equal, 
which happens with a small probability because $s$ is random. 

Formally, for almost all indices $i$ where $G_p$ performs some edit operation, 
the block around it $B=\{i-\frac{c}{r},i-\frac{c}{r}+1,\dots,i+\frac{c}{r}\}$ in $s$
(for a small constant $c>0$),
satisfies that $i$ is the only index in $B$ that $G_p$ edits (see Lemma~\ref{lem:tildeI}).
Next we show that in every \emph{optimal} alignment between $s$ and $y$, 
almost all these blocks contribute a cost of $1$. 
As otherwise, there is locally an alignment $M$ that aligns each index in $B$ 
to some character in $y$, 
while $G_p$ makes exactly one edit operation, say deletes $s[i]$. 
See for example Figure~\ref{fig:fig1}, 
where $M$ aligns $s[i]$ and $y[\bar{i}]$ 
whereas $A^p$ deletes $s[i]$.
In this case, $M$ and $A^p$ must disagree on at least $c/r$ indices 
(all indices either to the right or to the left of $i$ in $B$). 
In Figure~\ref{fig:fig1}, all $j\in[i-\frac{c}{r},i]$ satisfy $M[j]\neq A^p[j]$. 
The crux is that any pair of symbols in $s,y$ are chosen independently at random unless $A^p$ aligns their positions. 
Thus probability that in each of the $\frac{c}{r}$ pairs aligned by $M$, the two matched symbols will be equal is $(1/|\Sigma|)^{\frac{c}{r}}$. 
In the formal proof, we address several technical issues, 
like having not just one but many blocks, 
and possible correlations due to overlaps between different pairs,
which are overcome by a carefully crafted union bound. 

We further need to prove that the planted alignment is robust,
in the sense that, with high probability, 
every near-optimal alignment between $s$ and $y$ 
must "agree" with the planted alignment $A^p$ on all but a small fraction 
of the edit operations. 
Formally, we again consider a partition of $s$ into blocks containing exactly one planted edit operation,
and show that for almost all such blocks $B$, 
if $A^p$ maps $B$ to a substring $B'$ in $y$, 
that near-optimal alignment also maps $B$ to $B'$ 
(see Lemma~\ref{lem:unique-alignment}). 
To see this, suppose there is a near-optimal alignment that maps $B$ to $\bar{B}\neq B'$. Then following an argument similar to the above, we can show there are many indices in the block $B$ such that $A^p$ and $M$ disagree on them. 
Thus, in each such block, $M$ tries to match many pairs of symbols that are chosen independently at random, 
and therefore the probability that $M$ matches $B$ and $\bar{B}$ with a cost at most $1$ is small. 
Compared to Lemma~\ref{lem:tildeI}, an extra complication here 
is that now we allow $M$ to match $B$ and $\bar{B}$ with cost at most $1$ (and not only $0$), and in particular $\bar{B}$ can have three different lengths: $|B|,|B|-1,|B|+1$. 
Hence the analysis must argue separately for all these cases,  requiring a few additional ideas/observations.

After showing that a near-optimal alignment between a random string $s$ and $R_p(s)$ is almost unique (or robust), we use that fact to argue about the distance between $s$ and any approximate median of the traces. Let $s_1,s_2,s_3$ be three independent traces of $s$, generated by $R_p$. 
We can view $s_1$ as a uniformly random string, and $s_2,s_3$ are generated from $s_1$ by a insertion-deletion channel $R_q$ with a higher noise rate $q \approx 2p$. (See Section~\ref{sec:prob-model} for the details.) Hence, any near-optimal alignment between $s_1,s_2$ and $s_1,s_3$ "agree" with the planted alignment $A^q$ induced by $R_q$ (denoted by $A_{1,2}^q$ and $A_{1,3}^q$ respectively). 
Next, we consider the alignment $A_{1,2}$ from $s_1$ to $s_2$ via $s$, that we get by composing the planted alignment from $s_1$ to $s$ induced by $R_p$ (actually the inverse of the alignment from $s$ to $s_1$) with the planted alignment from $s$ to $s_2$ induced by $R_p$. Similarly, consider the alignment $A_{1,3}$ from $s_1$ to $s_3$ via $s$. Then we take any $(1+\epsilon)$-approximate median $\med$ of $\{s_1,s_2,s_3\}$. Consider an optimal alignment between $s_1,\med$, and $\med,s_2$, and $\med,s_3$. Use these three alignments to define an alignment $M_{1,2}$ from $s_1$ to $s_2$ via $\med$, and an alignment $M_{1,3}$ from $s_1$ to $s_3$ via $\med$. 
It is not hard to argue that both $A_{1,2}$ and $M_{1,2}$ are near-optimal alignments between $s_1,s_2$. Thus, both of them agree with the planted alignment $A_{1,2}^q$. Similarly, both $A_{1,3}$ and $M_{1,3}$ agree with the planted alignment $A_{1,3}^q$. Observe, $s_2,s_3$ are not independently generated from $s_1$ by $R_q$. The overlap between $A_{1,2}^q$ and $A_{1,3}^q$ essentially provides an alignment from $s_1$ to $s$. Again, using the robustness property of the planted alignment, this overlap between $A_{1,2}^q$ and $A_{1,3}^q$ agrees with the planted alignment from $s_1$ to $s$ by $A^p$ (actually the inverse of the alignment from $s$ to $s_1$). On the other hand, since $M_{1,2}$ agrees with $A_{1,2}^q$ and $M_{1,3}$ agrees with $A_{1,3}^q$, there is also a huge agreement between the overlap of $M_{1,2}, M_{1,3}$ and the overlap of $A_{1,2}^q,A_{1,3}^q$. The overlap between $M_{1,2}, M_{1,3}$ is essentially the optimal alignment from $s_1$ to $\med$ (that we have considered before). This in turn implies that there is a huge agreement between the optimal alignment from $s_1$ to $\med$ and the planted alignment from $s_1$ to $s$ by $R_p$. Hence, we can deduce that $\med$ and $s$ are the same in most of the portions, and thus has small edit distance. 
We provide the detailed analysis in Section~\ref{sec:unique-median}.

\subsection{Preliminaries}
\label{sec:prelims}

\paragraph{Alignments.}

For two strings $x,y$ of length $n$, 
an {\em alignment} is a function $A:[n]\to[n]\cup \{\bot\}$
that is monotonically increasing on the \emph{support} of $A$, 
defined as $\supp(A)\eqdef A^{-1}([n])$, 
and also satisfies $x[i]=y[A(i)]$ for all $i\in \supp(A)$.
An alignment is essentially a common subsequence of $x,y$, 
but provides the relevant location information. 
Define the \emph{length} (or support size) of the alignment as $\len(A)\eqdef |\supp(A)|$, 
i.e., the number of positions in $x$ (equivalently in $y$) that are matched by $A$.
Define the \emph{cost of $A$} to be the number of positions in $x$ and in $y$ that are not matched by $A$, 
i.e., $\cost(A)\eqdef 2 (n-\len(A))$. 
Let $ED(x,y)$ denotes the minimum cost of an alignment between $x,y$.

Given a substring $x'=x[i_1,i_2]$ of $x$, 
let $\ell_1\eqdef \min\set{k\in [i_1,i_2]\mid A(k)\neq \bot}$ 
and $\ell_2\eqdef \max\set{k\in [i_1,i_2]\mid A(k)\neq \bot}$,
be the first and last positions in the substring $x'$ that are matched by alignment $A$.
If the above is not well-defined, i.e., $A(k)=\bot$ for all $k\in [i_1,i_2]$, 
then by convention  $\ell_1=\ell_2=0$. 
Let $A(x')\eqdef y[A(\ell_1),A(\ell_2)]$ be the \emph{mapping} of $x'$ under $A$.
If $\ell_1=\ell_2=0$, then by convention $y'$ is an empty string.
Let $\mathcal{U}_{x'}\eqdef \{i_1\le k\le i_2; k\notin \supp(A)\}$ be the positions in $x'$ that are not aligned by $A$, 
and similarly let $\mathcal{U}_{y'}$ be the positions in $y'$ not aligned by $A$. 
These quantities are related because the number of matched positions in $x'$ is the same as in $y'$, 
giving us $|x'|-|\mathcal{U}_{x'}| = |y'|-|\mathcal{U}_{y'}|$. 
Define the \emph{cost of alignment $A$ on substring $x'$} to be 
$$\cost_A(x')\eqdef |\mathcal{U}_{x'}|+|\mathcal{U}_{y'}|.$$
%
By abusing the notation, sometimes we will also use $\cost_A([i_1,i_2])$ in place of $\cost_A(x')$. These definitions easily extend to strings of non-equal length, and even of infinite length. 

\begin{lemma}
\label{lem:subcost}
Given two strings $x,y$ and an alignment $A$, let $x_1,\dots,x_p$ be disjoint substrings of $x$. Then 
\begin{enumerate}
    \item $A(x_1),\dots,A(x_p)$ are disjoint substrings of $y$; and
    \item $\cost_A(x)\ge \sum_{i\in[p]}\cost_A(x_i)$
\end{enumerate}
\end{lemma}

\begin{proof}
The first claim directly follows from the fact that $x_1,\dots, x_p$ are disjoint and $A$ is monotonically increasing.

Since $x_1,\dots, x_p$ are disjoint,
also $\mathcal{U}_{A(x_1)},\dots,\mathcal{U}_{A(x_p)}$ are disjoint.
Similarly, since $A(x_1),\dots,A(x_p)$ are disjoint,
also $\mathcal{U}_{x_1},\dots,\mathcal{U}_{x_p}$ are disjoint. 
Therefore, $\cost_A(x)\ge \sum_{i\in[p]} (|\mathcal{U}_{x_1}|+|\mathcal{U}_{A(x_1)}|)$. Hence we can claim $\cost_A(x)\ge \sum_{i\in[p]}\cost_A(x_i)$.
\end{proof}

For an alignment $A:[n]\to[n]\cup \{\bot\}$, we define the inverse alignment $A^{-1} : [n] \to [n] \cup \{\bot\}$ as follows: For each $j \in [n]$, if $A(i)=j$ for some $i \in [n]$, set $A^{-1}(j)=i$; otherwise, set $A^{-1}(j)=\bot$.

We use the notation $\circ$ for composition of two functions. Composition of two alignments (or inverse of alignments) is defined in a natural way.


\paragraph{Approximate Median.}
Given a set $S \subseteq \Sigma^*$ and a string $y \in \Sigma^*$, we refer the quantity $\sum_{x \in S}\ED(y,x)$ by the \emph{median objective value} of $S$ with respect to $y$, denoted by $\obj(S,y)$. 

Given a set $S \subseteq \Sigma^*$, a \emph{median} of $S$ is a string $y^* \in \Sigma^*$ (not necessarily from $S$) such that $\obj(S,y^*)$ is minimized, i.e., $y^* = \argmin_{y \in \Sigma^*} \obj(S,y)$. We refer $\obj(S,y^*)$ by $\opt(S)$. Whenever it will be clear from the context, for brevity we will drop $S$ from both $\obj(S,y)$ and $\opt(S)$. We call a string $\tilde{y}$ a \emph{$c$-approximate median}, for some $c>0$, of $S$ iff $\obj(S,\tilde{y})\le c \cdot \opt(S)$.

\section{Probabilistic Generative Model}
\label{sec:prob-model}
 
Let us first introduce a probabilistic generative model. 
For simplicity, our model is defined using infinite-length strings, 
but our algorithmic analysis will consider only a finite prefix of each string. 
Fixing a finite alphabet $\Sigma$, we denote by $\Sigma^{\N}$ 
the set of all infinite-length strings over $\Sigma$. 
We write $x\odot y$ to denote the concatenation of two finite-length strings $x$ and $y$. 

We actually describe two probabilistic models that are equivalent.
The first model $R_p$ is just the insertion-deletion channel 
mentioned in Section~\ref{sec:intro}. 
These models are given an arbitrary string $x$ (base string) 
to generate a random string $y$ (a trace),
but in our intended application $x$ is usually a random string. 
The second model $G_p$ consists of two stages, 
first "planting" an  alignment between two strings, 
and only then placing random symbols (accordingly).
This is more convenient in the analysis, 
because we often want to condition on the planted alignment 
and rely on the randomness in choosing symbols.

\paragraph{Model $R_p(x)$.}
Given an infinite-length string $x \in \Sigma^{\N}$ and a parameter $p\in [0,1]$.
Consider the following random procedure:
 \begin{enumerate} \compactify
 \item Initialize $i=1$. (We use $i$ to point to the current index positions of the input string.) Also, initialize an empty string $Out$.
 \item Do the following independently at random:
 \begin{enumerate} \compactify
 \item With probability $1-p$, set $Out\gets Out \odot x[i]$ and increment $i$. (Match $x[i]$.)
 \item With probability $p/2$, increment $i$. (Delete $x[i]$.)
 \item With probability $p/2$, choose independently uniformly at random a character $a\in\Sigma$ and set $Out \gets Out \odot a$. (Insert a random character.)
 \end{enumerate}
 \end{enumerate}
We call this procedure $R_p$, and denote the randomized output string $Out$ by $R_p(x)$. 

\paragraph{Model $G_p(x)$.}
This model first provides a randomized mapping (alignment) 
$A^p:\N\to \N \cup \{\bot\}$, 
and then uses this alignment (and $x$) to generate the output string. 
First, given a parameter $p \in [0,1]$, 
consider the following random procedure to get a mapping $A^p$:
 \begin{enumerate} \compactify
 \item Initialize $i=1$ and $j=1$. (Indices of current positions in the input and output strings, respectively.)
 \item Do the following independently at random:
 \begin{enumerate} \compactify 
 \item With probability $1-p$, set $A^p(i)\gets j$ and increment both $i$ and $j$. (Match $x[i]$.)
 \item With probability $p/2$, set $A^p(i)\gets \bot$ and increment $i$. (Delete $x[i]$.)
 \item With probability $p/2$, increment $j$. (Insertion.)
 \end{enumerate} 
 \end{enumerate}
 Next, use this  $A^p$ 
 and a given string $x \in \Sigma^{\N}$ 
 to generate a string $y\in \Sigma^{\N}$ as follows. 
 For each $j \in \N$,
 \begin{itemize} \compactify 
 \item If there is $i\in\N$ with $A^p(i)=j$ 
 then set $y[j]\gets x[i]$. (Match $x[i]$.)
 \item Otherwise, choose independently uniformly at random a character $a\in\Sigma$ and set $y[j]\gets a$. (Insert a random character.)
 \end{itemize}
We denote by $G_p(x)$ the randomized string $y$ generated as above. By construction, $A^p$ is an alignment between $x$ and $G_p(x)$. 

\paragraph{Basic Properties.}
We claim next that $R_p$ and $G_p$ are equivalent, which is useful because we find it more convenient to analyze $G_p$.
We use $X_1 \deq X_2$ to denote that two random variable $X_i \sim D_i$, $i \in \{1,2\}$ have equal distribution, i.e., $D_1 = D_2$.
The next two propositions are immediate. 
 
\begin{proposition}
\label{prop:identical-model}
For every string $x \in \Sigma^{\N}$ and $p \in [0,1]$, 
we have $R_p(x) \deq G_p(x)$.
\end{proposition}

\begin{proposition}[Transitivity]
\label{prop:dist-transitivity}
For every $x \in \Sigma^{\N}$ and $p \in [0,1]$,
let 
\begin{equation} \label{eq:qofp}
  q(p)\eqdef\frac{p(4-3p)}{2-p^2}.
\end{equation}
Then $G_p(G_p(x)) \deq G_{q(p)}(x)$.
\end{proposition}

\paragraph{Additional Properties (Random Base String).} 
Let $U$ be the uniform distribution over strings $x\in\Sigma^{\N}$, i.e., each character $x[i]$ is chosen uniformly at random and independently from $\Sigma$. 
We now state two important observations regarding the process $G_p$. 
The first one is a direct corollary of Proposition~\ref{prop:dist-transitivity}. %
The second observation follows because the probability of insertion is the same as that of deletion at any index in the random process $G_p$. 

\begin{corollary}[Transitivity]
\label{cor:dist-pair-transitivity}
Let $p \in [0,1]$ and let $q(p)$ be as in~\eqref{eq:qofp}. 
Draw a random string $X\sim U$, and use it to draw $Y\sim G_p(G_p(X))$ and $Z\sim G_{q(p)}(X)$. Then $(X,Y) \deq (X,Z)$.
\end{corollary}

\begin{proposition}[Symmetry]
\label{prop:dist-symmetry}
Let $p \in [0,1]$. Draw a random string $X\sim U$ and use it to draw another string $Y\sim G_p(X)$. Then $(X,Y) \deq (Y,X)$.
\end{proposition}

\section{Robustness of the Insertion-Deletion Channel}
\label{sec:unique-edit}

In this section we analyze the finite-length version of 
our probabilistic model $G_p$ (from Section~\ref{sec:prob-model}),
which gets a random string $x\in\Sigma^n$ 
and generates from it a trace $y$.
We provide a high-probability estimates 
for the cost of the planted alignment $A^p$ between $x$ and $y$ (Lemma~\ref{lem:number-edit-operation}), 
and for the optimal alignment (i.e., edit distance) between the two strings
(Lemmas~\ref{lem:editbound-small-alphabet}
and~\ref{lem:editbound-large-alphabet}).
It follows that with high probability 
the planted alignment $A^p$ is near-optimal. 
We then further prove that the planted alignment is robust,
in the sense that, with high probability, 
every near-optimal alignment between the two strings 
must "agree" with the planted alignment $A^p$ on all but a small fraction 
of the edit operations 
(Lemmas~\ref{lem:unique-alignment} and~\ref{lem:unique-large-alignment}). 

Assume henceforth that $x\in\Sigma^n$ is a random string $x$,
and given a parameter $p>0$, 
generate from it a string $y$ by the random process $G_p$ described 
in Section~\ref{sec:prob-model},
denoting by $A^p_{x,y}$ the random mapping used in this process.
A small difference here is that now $x$ has finite length,
but it can also be viewed as an $n$-length prefix of an infinite string.
Similarly, now $y$ has finite length and is obtained by applying $G_p$ on $x[1,n]$,
and it can be viewed also as a finite prefix of an infinite string. 

Let $\mathcal{I}^{A^p_{x,y}}$ be the set of indices $i\in[n]$, 
for which process $G_p$ performs at least one insert/delete operation
after (not including) $x[i-1]$ and up to (including) $x[i]$
(i.e., inserting at least one character between $x[i-1],x[i]$, or deleting $x[i]$, or both). 
This information can clearly be described using $A^p_{x,y}$ alone 
(independently of $x$ and of the symbols chosen for insertions to $y$ in the second stage of process $G_p$); 
we omit the formal definition. 
When clear from the context, we shorten $\mathcal{I}^{A^p_{x,y}}$ to $\mathcal{I}$.

\begin{lemma}
\label{lem:number-edit-operation}
For every $i\in[n]$, the probability that $i\in\mathcal{I}$ is
\begin{equation} \label{eq:rofp}
  r=r(p)\eqdef \sum_{k=1}^{\infty} (2-p) (p/2)^k=p.
\end{equation}
For all  $\epsilon \in [r,1]$, we have
$\Pr [ |\mathcal{I}| \notin (1\pm\epsilon) r n ] 
    \leq 2e^{-{\epsilon^2 r n}/{3}}$. 
\end{lemma}

\begin{proof} 
For every $i\in[n]$, the probability that $G_p$ performs $k\ge1$ edit operations after $x[i-1]$ and up to $x[i]$ is $(p/2)^k+(p/2)^k(1-p)=(p/2)^k(2-p)$,
where the first summand represents $k-1$ insertions and one deletion, and the second summand represents $k$ insertions and no deletion. 
Thus $r=\Pr[i\in\mathcal{I}] = \sum_{k=1}^{\infty} (2-p) (p/2)^k = p$.

For every $i\in[n]$, the probability it appears in $\mathcal{I}$ is $r$. 
Then $\EX[|\mathcal{I}|]=r\cdot n$. 
These events are independent, hence by Chernoff's bound, the probability that $|\mathcal{I}|\notin (1\pm\epsilon)r n$ is at most $2e^{-{\epsilon^2 r n}/{3}}$.
\end{proof}

As shown in~\eqref{eq:rofp}, $r\eqdef \sum_{k=1}^{\infty} (2-p) (p/2)^k=p$. Hence from now on we replace $r$ by $p$. 
Given $\epsilon\in[p,1]$ and $i\in [n]$, 
we consider the event $\mathcal{S}_\epsilon(i)$, which informally means 
that process $G_p$ makes a single "well-spaced" edit operation at position $i$,
i.e., there is an edit operation at position $i$
and no other edit operations within $(\tfrac{2\epsilon}{p}$ positions away from $i$.
To define it formally, we separate it into two cases, an insertion and a deletion.
Observe that these events depend on $A^p$ alone. 
Let $\mathcal{S}^{del}_\epsilon(i)$ be the event that 
\begin{enumerate} \compactify
  \item $A^p(i+1)=A^p(i-1)+1$ (thus $A^p(i)=\bot$); and
  \item for all $j\in[2,\frac{2\epsilon}{p}]$,  
  we have $A^p(i+j)=A^p(i+j-1)+1$
  and $A^p(i-j)=A^p(i-j+1)-1$
  (in particular, they are not $\bot$). 
\end{enumerate}
Similarly, let $\mathcal{S}^{ins}_\epsilon(i)$ be the event that 
\begin{enumerate} \compactify 
  \item $A^p(i)=A^p(i-1)+2$ (thus no index is mapped to $A^p(i-1)+1$); 
  \item for all $j\in[1,\frac{2\epsilon}{p}]$,  
  we have $A^p(i+j)=A^p(i+j-1)+1$; and
  \item for all $j\in[2,\frac{2\epsilon}{p}]$,  
  and $A^p(i-j)=A^p(i-j+1)-1$.
\end{enumerate}
Now define 
the set of indices for which any of these two events happens 
\[
  \mathcal{\tilde{I}}_{\epsilon}^{A^p_{x,y}}
  \eqdef
  \set{ i\in[n] \mid 
  \text{event $\mathcal{S}_\epsilon(i)\eqdef \mathcal{S}^{del}_\epsilon(i) \cup \mathcal{S}^{ins}_\epsilon(i)$ occurs} }.
\]
When clear from the context, we shorten $\mathcal{\tilde{I}}_{\epsilon}^{A^p_{x,y}}$ to $\tilde{\mathcal{I}}$. 

\begin{lemma}
\label{lem:tildeI} 
For every $\epsilon \ge p$, 
we have 
$\Pr[ |\mathcal{\tilde{I}}| \le (1-5\epsilon)p n ]
  \leq e^{-{\epsilon^2 p^2 n}/{2}}$.
\end{lemma}

\begin{proof}
For an index $i\in[n]$, 
define the random variable $X_i\in\set{0,1}$ to be an indicator for the union event  
$\mathcal{S}^{del}_\epsilon(i) \cup \mathcal{S}^{ins}_\epsilon(i)$.
Observe that each of the two events, 
$\mathcal{S}^{del}_\epsilon(i)$ and $\mathcal{S}^{ins}_\epsilon(i)$, 
occurs with probability
$\frac{p}{2}(1-p)^{4\epsilon/p}/(1-p/2)$,
and these events are disjoint.
Hence, $\Pr[X_i=1]= p(1-p)^{4\epsilon/p}/(1-p/2)\ge p(1-p)^{4\epsilon/p} \ge p(1-4\epsilon)$ 

Next we prove a deviation bound for the random variable 
$X\eqdef \sum_{i\in[n]} X_i = |\mathcal{\tilde{I}}|$,
which has expectation is $\EX[X] \geq (1-4\epsilon)pn$. 
Observe that $\set{X|X_1,\ldots,X_i}_{i=1}^n$ is a Doob Martingale,
and let us apply the method of bounded differences.
Revealing $X_i$ (after $X_1,\ldots,X_{i-1}$ are already known)
might affect the value of $X_j$'s for $j< i+\frac{4\epsilon}{p}$, 
but by definition their sum is bounded 
$\sum_{i\le j <i+\frac{4\epsilon}{p}} X_j \leq 2$,
while the other $X_j$'s are independent of $X_i$ 
hence the expectation of $\sum_{j\geq i+\frac{2\epsilon}{p}} X_j$ by revealing it. 
Together, we see that 
$ | \EX[X|X_1,\dots,X_i] - \EX[X|X_1,\dots,X_{i-1}] | \le 2$,
and therefore by Azuma's inequality,  
$
  \Pr[X \leq (1-5\epsilon)pn]
  \le  \Pr[X \leq \EX[X] -\epsilon pn]
  < e^{-2\epsilon^2 p^2 n^2/(4n)}
  = e^{-\epsilon^2 p^2 n/2}
  $.
\end{proof}

\paragraph{Edit Distance (Optimal Alignment) between $x,y$.}

For each $i\in \mathcal{\tilde{I}}$, define a window $W_\epsilon^i=[i-\frac{\epsilon}{p},i+\frac{\epsilon}{p}]$.

\begin{lemma}
\label{lem:editbound-small-alphabet}
For every $\epsilon \in [15p\log \frac{1}{p},\frac16]$, 
we have
$\Pr[ \ED(x,y) < (1-6\epsilon)pn ] \leq 2e^{-\epsilon^2p^2 n/2}$.
\end{lemma}
At a high level,
our proof avoids a direct union bound over all low-cost potential alignments, because there are too many of them.
Instead, we introduce a smaller set of basic events 
that "covers" all these potential alignments,
which is equivalent to carefully grouping the potential alignments 
to get a more "efficient" union bound. 

\begin{proof}
We assume henceforth that $A^p$ (the alignment from process $G_p$) is known 
and satisfies $|\mathcal{\tilde{I}}| > (1-5\epsilon)p n$,
which occurs with high probability by Lemma~\ref{lem:tildeI}. 
In other words, we condition on $A^p$ and proceed with a probabilistic analysis 
based only on the randomness of $x$ and of the characters inserted into $y$. 

Our plan is to define basic events $\mE_{S,\bar{S}}$
for every two subsets $S,\bar{S}\subset[n]$ of the same size $\ell=|S|=|\bar{S}|$,
representing positions in $x$ and in $y$, respectively. 
We will then show that our event of interest is bounded by these events   
\begin{align} \label{eq:editbound1}
  \BIGset{ \ED(x,y) < (1-6\epsilon)pn }
  \subseteq  
  \bigcup_{S,\bar{S} \mid \ell= \epsilon pn} \mE_{S,\bar{S}} ,
\end{align}
and bound the probability of each basic event by
\begin{align} \label{eq:editbound2}
  \Pr[ \mE_{S,\bar{S}} ] 
  \leq 
  |\Sigma|^{ -\epsilon\ell/(3p) } .
\end{align} 
The proof will then follow easily using a union bound and a simple calculation.

To define the basic event $\mE_{S,\bar{S}}$, we need some notation.  
Write $S=\{i_1,i_2,\dots,i_\ell\}$ in increasing order, 
and similarly $\bar S=\{\bar{i}_1,\bar{i}_2,\dots,\bar{i}_\ell\}$,
and use these to define $\ell$ blocks in $x$ and in $y$, namely, 
$B_{i_j}=x[i_j-\tfrac{\epsilon}{p} , i_j+\tfrac{\epsilon}{p}]$
and 
$\bar{B}_{i_j}=y[\bar{i}_j-\tfrac{\epsilon}{p} , \bar{i}_j+\tfrac{\epsilon}{p}]$. 
Notice that all the blocks are of the same length $1+2\tfrac{\epsilon}{p}$.
Now define $\mE_{S,\bar{S}}$ to be the event that
(i) $S\subseteq \tI$;%
\footnote{This implies that the blocks $B_{i_1},\dots,B_{i_\ell}$ in $x$ are disjoint.}
(ii) the blocks $\bar{B}_{\bar{}i_1},\dots,\bar{B}_{\bar{i}_\ell}$ in $y$ are disjoint;
and
(iii) each block $B_{i_j}$ in $x$ is equal to its corresponding block $\bar{B}_{\bar{i}_j}$ in $y$.
Notice that conditions (i) and (ii) actually depend only on $A^p$,
and thus can be viewed as restrictions on the choice of $S,\bar{S}$ in~\eqref{eq:editbound1};
with this viewpoint in mind, we can simply write 
\[
  \mE_{S,\bar{S}} 
  \eqdef 
  \set{ B_{i_1}=\bar{B}_{\bar{i}_1}, \ldots, B_{i_\ell}= \bar{B}_{\bar{i}_\ell} } .
\]

    We proceed to prove~\eqref{eq:editbound1}. 
Suppose there is an alignment $M$ from $x$ to $y$ with $\cost(M) < (1-6\epsilon)pn$,
and consider its cost around each position $i\in\tI$, 
namely, $\cost_M[i-\tfrac{\epsilon}{p} , i+\tfrac{\epsilon}{p}]$.
These intervals in $x$ are disjoint (by definition of $\tI$), 
and thus by Lemma~\ref{lem:subcost}, 
\[
  \sum_{i\in\tI} \cost_M(x[i-\tfrac{\epsilon}{p} , i+\tfrac{\epsilon}{p}])
  \le \cost(M)
  < (1-6\epsilon)pn.  
\]
Let $S\subset\tI$ include (the indices of) the summands equal to $0$.
Each other summand contributes at least $1$, 
thus $|\tI|-|S| = |\tI\setminus S|\cdot 1 < (1-6\epsilon)pn$ 
and by rearranging 
$
  |S| 
  > |\tI| - (1-6\epsilon)pn 
  > \epsilon pn 
$. 
To get the exact size $|S|=\epsilon pn$, 
we can replace $S$ with an arbitrary subset of it of the exact size. 
Now define $\bar S = \set{ M(i) \mid i\in S}$. 
It is easy to verify that the event $\mE_{S,\bar{S}}$ holds. 
Indeed, each $i\in S$ 
satisfies $\cost_M[i-\tfrac{\epsilon}{p} , i+\tfrac{\epsilon}{p}] = 0$,
which implies $M(i)\neq\bot$, and thus $|\bar{S}|=|S|$. 
Moreover, the block $x[i-\tfrac{\epsilon}{p} , i+\tfrac{\epsilon}{p}]$ in $x$ 
is equal to the corresponds block in $y$, 
and these blocks in $y$ are disjoint. 
This completes the proof of~\eqref{eq:editbound1}.

Next, we prove~\eqref{eq:editbound2}. 
Fix $S,\bar S\subset[n]$ of the same size $\ell$, 
and assume requirements (i) and (ii) hold (otherwise, the probability is 0). 
Let $B_{i_{j}}$ and $\bar{B}_{\bar{i}_j}$ be the corresponding blocks in $x$ and in $y$. 
Consider for now a given $j\in[\ell]$. 
The requirement $B_{i_{j}} = \bar{B}_{\bar{i}_j}$ 
means that for all $t\in\set{ -\tfrac{\epsilon}{p},\ldots,0,\ldots,+\tfrac{\epsilon}{p} }$ 
we require $x[i_j + t] = y[\bar{i}_j+t]$. 
The issue is that $x$ and $y$ are random but correlated through $A^p$;
in particular, the symbols $x[i_j + t]$ and $y[\bar{i}_j+t]$ are chosen independently at random 
unless $A^p$ aligns their positions, i.e., $A^p(i_j +t) = \bar{i}_j +t$. 
The key observation is that this last event cannot happen for both $t=-1$ and $t=1$,
because in that case, 
$A^p(i_j +1) - A^p(i_j -1) = \bar{i}_j+1 - (\bar{i}_j-1) = 2$; 
however, $i_j\in\tI$ implies that 
$A^p$ has exactly one edit operation (insertion or deletion)
in the interval $[i_j-1,i_j+1]$ (and not at its endpoints), 
thus $A^p(i_j +1) - A^p(i_j -1)\in \set{1,3}$. 
Assume first that $A^p(i_j +t) \neq \bar{i}_j +t$ for $t=1$. 
Then the same must hold also for all $t=2,\ldots,\tfrac{\epsilon}{p}$;
indeed, we again use that $i_j\in\tI$,
which implies that $A^p$ has no edit operations near position $i_j$,
thus $A^p(i_j +t) = A^p(i_j +1) + (t-1) \neq \bar{i}_j +1 + (t-1)$. 
%
The argument for $t=-1$ is similar, 
and we conclude that the requirement $B_{i_{j}} = \bar{B}_{\bar{i}_j}$ 
encompasses at least $\tfrac{\epsilon}{p}$ requirements of the form
$x[i_j + t] = y[\bar{i}_j+t]$ 
where these two positions are not aligned by $A^p$, 
and thus these two symbols are chosen independently at random. 

The above argument applies to every $j\in[\ell]$,
yielding overall at least $\ell\cdot \tfrac{\epsilon}{p}$ requirements of the form  
$x[i_j + t] = y[\bar{i}_j+t]$,
where these two symbols are chosen independently at random. 
Observe that each $y[\bar{i}_j+t]$ is 
either a character $x[t']$ (for $t'$ arising from $A^p$) or completely independent. 
Since each character of $x$ appears in at most $2$ requirements (once on each side),
we can extract a subset of at one-third of the requirements
such that the positions in $x$ appearing there are all distinct,
and thus the events are independent.%
\footnote{To see this, consider an auxiliary graph whose a vertex for each character $x[t]$,
and connect two by an edge if they appear in the same constraint.
Since every vertex has degree at most $2$, a greedy matching contains at least one third of the edges.
}
We overall obtain at least $\tfrac13 \ell\cdot \tfrac{\epsilon}{p}$ requirements,
each occurring independently with probability $1/|\Sigma|$,
and thus
\[
  \Pr[ \mE_{S,\bar{S}} ] 
  \leq 
  |\Sigma|^{ -\epsilon\ell/(3p) } .
\]

Finally, we are in position to prove the lemma. 
Combining \eqref{eq:editbound1} and \eqref{eq:editbound2} and a union bound
\begin{align*}
  \Pr[ \ED(x,y) < (1-6\epsilon)pn ]
  &\leq \binom{n}{\ell}^2 \cdot |\Sigma|^{ -\epsilon\ell/(3p) } 
  \leq \Big( \frac{n e}{\ell} \Big)^{2\ell} \cdot 2^{ -\epsilon\ell/(3p) } 
  =  \Big( \frac{e}{\epsilon p} \Big)^{2\epsilon pn} \cdot 2^{ -\epsilon^2 n/3 } 
  \\
  &\leq (p^2)^{-2\epsilon pn} \cdot 2^{ -\epsilon (15p\log(1/p)) n/3 } 
  \leq p^{-4\epsilon pn + 5\epsilon pn}
  \leq p^{\epsilon pn} .
\end{align*}

Recall that this was all conditioned on $A^p$,
which had error probability at most $e^{-{\epsilon^2 p^2 n}/{2}}$ (by Lemma~\ref{lem:tildeI}),
and now Lemma~\ref{lem:editbound-small-alphabet} follows by a union bound. 
\end{proof}


A similar bound holds for even smaller values of $\epsilon$, 
provided that the alphabet size is large. 
The proof is the same, except for the final calculation.  

\begin{lemma}
\label{lem:editbound-large-alphabet} 
Suppose $|\Sigma|\ge (\frac{1}{p})^{15}$.
Then for every $\epsilon \in [p,\frac16]$, 
we have 
$\Pr[ \ED(x,y) < (1-6\epsilon)pn ] \leq 2e^{-\epsilon^2p^2 n/2}$.
\end{lemma}

Following an argument similar to the proof of Lemma~\ref{lem:editbound-small-alphabet} we can make the following claim.

\begin{lemma}
\label{lem:zerocostbound} 
Let $\epsilon\in[15p\log \frac{1}{p},\frac{1}{6}]$. Then with probability at least $1-2e^{-\epsilon^2p^2 n/2}$, every alignment $M$ between $x,y$ satisfies $|\{i\in \mathcal{\tilde{I}}\mid \cost_M(x[i-\frac{\epsilon}{r},i+\frac{\epsilon}{r}])=0\}|\le 6\epsilon p n$.
\end{lemma}

\paragraph{Near-Optimal Alignments between $x,y$.}

Given $\epsilon>0$, a potential alignment $M$ between $x,y$, 
and an index $i\in [n]$,
define the event 
\begin{equation} \label{eq:EspM1} 
 \mathcal{E}_\epsilon^M(i)\eqdef 
\begin{cases}
 A^p(i-\tfrac{\epsilon}{p})= \min \{M(k)\neq \bot 
   \mid k\in[i-\frac{\epsilon}{p},i+\frac{\epsilon}{p}]\}; 
 \text{ and }\\
 A^p(i+\tfrac{\epsilon}{p})=\max \{M(k)\neq \bot 
   \mid k\in [i-\frac{\epsilon}{p},i+\frac{\epsilon}{p}]\}.
\end{cases}
\end{equation}

By convention,  $\mathcal{E}_\epsilon^M(i)$ is \emph{not} satisfied
if the minimization/maximization is over the empty set (because $M(k)=\bot$ for all relevant $k$).
We will only use it for $i\in\tI$,
in which case both 
$A^p(i-\frac{\epsilon}{p}), A^p(i+\frac{\epsilon}{p}) \neq \bot$. 
Intuitively, this event means that $A^p$ and $M$ agree on the block boundaries;
for example, in the simpler case where all relevant $M(k)\neq\bot$,
this event simply means that 
$A^p(i-\tfrac{\epsilon}{p})=M(i-\tfrac{\epsilon}{p})$ and 
$A^p(i+\tfrac{\epsilon}{p})=M(i+\tfrac{\epsilon}{p})$.

%
Denote the set of indices where 
the event $\mathcal{E}_\epsilon^M(i)$ occurs 
and the cost of $M$ over substring $x[i-\tfrac{\epsilon}{p},i+\tfrac{\epsilon}{p}]$ is $1$, 
by 
\[
  \mathcal{\tilde{I}}_{\epsilon,M}^{A^p_{x,y}}
  \eqdef
  \set{ i\in\mathcal{\tilde{I}} \mid
    \text{ event $\mathcal{E}_\epsilon^M(i)$ occurs and $\cost_M([i-\tfrac{\epsilon}{p},i+\tfrac{\epsilon}{p}]) = 1$ }
  }.
\]
When clear from the context, we shorten $\mathcal{\tilde{I}}_{\epsilon,M}^{A^p_{x,y}}$ to $\mathcal{\tilde{I}}_M$.

\begin{lemma}
\label{lem:unique-alignment}
Let $\epsilon \in [42p\log \frac{1}{p},\frac{1}{6}]$ and $p\le \delta \le \epsilon$.
Then with probability at least $1-4e^{-\frac{\epsilon^2p^2n}{2}}$,
every alignment $M$ between $x,y$ with $\cost(M)\le (1+\delta)p n$ 
satisfies $|\mathcal{\tilde{I}}_M|\ge (1-23\epsilon-\delta)p n$. 
\end{lemma}

At a high level, the proof follows the outline of Lemma~\ref{lem:editbound-small-alphabet},
and avoids a direct union bound over all (relevant) potential alignments, 
because there are too many of them.
Instead, we introduce a smaller set of basic events 
that "covers" all these potential alignments.
Specifically, we show that an alignment $M$ violating the above 
gives rise to a large set $\tilde{S}$ of indices $i\in\tI$, 
where a corresponding block $B_i$ in $x$ is matched by $M$ 
to a block $\bar{B}_i$ in $y$ with cost at most $1$,
and event $\mathcal{E}^M_\epsilon(i)$ is not satisfied. 
It then remains to show that the probability that such a large set $\tilde{S}$ exists is very small. 
A crucial difference is that the cost of matching $B_i$ to $\bar{B}_i$ is at most $1$ here
(compared with $0$ in the proof of Lemma~\ref{lem:editbound-small-alphabet}).
It makes the notation more cumbersome, 
e.g, the length of $\bar{B}_i$ can be either $|B_i|$, $|B_i|-1$, or $|B_i|+1$,
and the analysis more elaborate with additional cases that require new technical ideas. 
We provide this proof in Appendix~\ref{app:unique-alignment}.


A similar bound holds for even smaller values of $\epsilon$, 
provided that the alphabet size is large. 

\begin{lemma} \label{lem:unique-large-alignment}
Suppose $|\Sigma|\ge (\frac{1}{p})^{42}$. Then for every $\epsilon \in[p,\frac{1}{6}]$ and every $p\le\delta \le \epsilon$, with probability at least $1-4e^{-\frac{\epsilon^2p^2n}{2}}$,
every alignment $M$ between $x,y$ with $\cost(M)\le (1+\delta)p n$ 
satisfies $|\mathcal{\tilde{I}}_M|\ge (1-23\epsilon-\delta)p n$. 
\end{lemma}

\section{Robustness of Approximate Median}
\label{sec:unique-median}
In this section, we consider the (approximate) median string problem on a set of strings generated by our probabilistic model $G_p$ (from Section~\ref{sec:prob-model}). For a random (unknown) string $s \in \Sigma^n$, $G_p$ generates a set $S = \{s_1,s_2,\cdots,s_m\}$ of independent traces of $s$. We show that with high probability, any $(1+\epsilon)$-approximate median of $S$ must be close (in edit distance) to the unknown string $s$. In other words, any $(1+\epsilon)$-approximate median must "agree" with the unknown string $s$ in most of the portions. It is true even when $m=3$. In this section, we state the results and the proofs by considering $m=3$. In particular, we prove Theorem~\ref{thm:median-trace-close}. At the end of the section, we remark on why such result with three traces also directly provides a similar result for any $m>3$ traces. Another way to interpret this result is the following. Suppose we take a set of three traces and find its $(1+\epsilon)$-approximate median. Then if we add more traces in the set,  its $(1+\epsilon)$-approximate median does not change by much. So in some sense, $(1+\epsilon)$-approximate median is robust in the case of average-case traces.

For the purpose of the analysis, we start by considering infinite length strings (as in Section~\ref{sec:prob-model}), and then later we will move to the finite-length versions. Recall that $U$ denotes the uniform distribution over strings $x\in\Sigma^{\N}$, i.e., each character $x[i]$, for $i \in \N$, is chosen uniformly at random and independently from $\Sigma$. Consider a parameter $p \in (0,0.001)$ and define $q\eqdef\frac{p(4-3p)}{2-p^2}$. (Note, $q=2p-\Theta(p^2)$.) Then consider the following two processes:
\begin{itemize}
\item \textbf{Process 1}: Draw a string $s$ from $U$. Then draw three strings $s_1,s_2,s_3$ independently from $G_p(s)$. Output the tuple $(s,s_1,s_2,s_3)$.
\item \textbf{Process 2}: Draw a string $x_1$ from $U$. Then draw $\bar{x}$ from $G_p(x_1)$ (and denote the corresponding alignment function by $A_{1,\bar{x}}^p$). Finally, draw $x_2,x_3$ independently from $G_p(\bar{x})$ (and denote the corresponding alignment functions by $A_{\bar{x},2}^p, A_{\bar{x},3}^p$ respectively). Output the tuple $(\bar{x},x_1,x_2,x_3)$.
\end{itemize}

As an immediate corollary of Proposition~\ref{prop:dist-symmetry}, we know that the distributions on $(s,s_1)$ and $(\bar{x},x_1)$ are the same. So we conclude the following about the above two processes.

\begin{claim}
\label{clm:equivalence}
The probability distributions on $(s,s_1,s_2,s_3)$ and $(\bar{x},x_1,x_2,x_3)$,
the tuples generated by Process 1 and Process 2 respectively, 
are identical.
\end{claim}

We want to investigate the property of an approximate median of the strings generated through Process 1. Due to the above claim, instead of considering the strings $s_1,s_2,s_3$ from now on we focus on $x_1,x_2,x_3$ generated through Process 2. By Proposition~\ref{prop:dist-transitivity}, both $x_2$ and $x_3$ can be viewed as strings drawn from $G_q(x_1)$. Let us use the notations $A_{1,2}^q$ and $A_{1,3}^q$ to denote the alignment functions produced by the random process $G_q$ while generating $x_2$ and $x_3$ respectively, from $x_1$. We want to emphasize that the process $G_q$ is considered solely for the purpose of the analysis.

Next, we use the alignments $A_{1,\bar{x}}^p,$ $A_{\bar{x},2}^p$ (and $A_{\bar{x},3}^p$) to define an alignment between $x_1,x_2$ (and $x_1,x_3$) via $\bar{x}$. Let $A_{1,\bar{x},2}^p$ and $A_{1,\bar{x},3}^p$ denote $A_{\bar{x},2}^p \circ A_{1,\bar{x}}^p$ and $A_{\bar{x},3}^p \circ A_{1,\bar{x}}^p$ respectively. (See Section~\ref{sec:prelims} for the definition of the notation $\circ$.)

\paragraph{Median of $n$-length prefixes of $x_1,x_2,x_3$. }
So far in this section we have talked about infinite length strings. From now on we restrict ourselves to the the $n$-length prefixes of $x_1,x_2$ and $x_3$ denoted by $x_1[1,n], x_2[1,n]$ and $x_3[1,n]$ respectively. By abusing the notations, we simply use $x_1,x_2$ and $x_3$ to also denote $x_1[1,n], x_2[1,n]$ and $x_3[1,n]$ respectively. Also, we consider the ($n$-length) restriction of all the alignment functions (defined so far) accordingly. Again, for simplicity, we use the same notations to refer to these restricted alignment functions.

Now, we consider the (approximate) median string problem on the set $S = \{x_1,x_2,x_3\}$. Recall, for any string $y$, $\obj(S,y) \eqdef \sum_{k=1}^{3}\ED(x_k,y)$, and $\opt(S)=\min_{y\in \Sigma^*} \obj(S,y)$. Since throughout this section, $S = \{x_1,x_2,x_3\}$, to simplify the notations, we drop $S$ from both $\obj$ and $\opt$. The main result of this section is the following.
\begin{theorem}
\label{thm:median-trace-close-alternate}
For a large enough $n \in \N$ and a noise parameter $p \in (0,0.001)$, 
let $\bar{x}, x_1,x_2$ and $x_3$ be the $n$-length prefixes of the strings generated by Process 2. If $\med$ is a $(1+\epsilon)$-approximate median of $S=\{x_1,x_2,x_3\}$ for $\epsilon \in [110 p \log (1/p),1/6]$, 
then $\Pr [\ED(\bar{x},\med) \le 195\epsilon\cdot \opt(S) ] \geq 1-e^{-\log^2 n}$.
\end{theorem}
We would like to emphasize that (for the simplicity in the analysis) we have made no attempt to optimize the constants. By a more careful analysis, both the range of $p$ and the constant involved in the bound of $\ED(\bar{x},\med)$ could be improved significantly. The above theorem together with Claim~\ref{clm:equivalence} immediately gives us Theorem~\ref{thm:median-trace-close}. Note, in Theorem~\ref{thm:median-trace-close}, we do not have any length restrictions on the traces. On the other hand, the above theorem considers $\bar{x}, x_1,x_2$ and $x_3$ to of of length $n$. However, by a standard application of Chernoff-Hoeffding bound, it suffices to restrict ourselves to the $(n-\sqrt{n} \log n)$-length prefixes of all the traces (of Theorem~\ref{thm:median-trace-close}). Then we can apply the above theorem over them, to get Theorem~\ref{thm:median-trace-close}.

Before proving Theorem~\ref{thm:median-trace-close-alternate}, we make a few observations on the ($n$-length restricted) alignments between $x_1$ and $x_2$ ($x_1$ and $x_3$). Consider an $\epsilon \in [110 p \log (1/p),1/6]$. Let $\med$ be an (arbitrary) $(1+\epsilon)$-approximate median of $x_1,x_2$ and $x_3$. 

Let $M_{1,\m}, M_{\m,x}$ and $M_{\m,x}$ be (arbitrary) optimal alignment from $x_1$ to $\med$, $\med$ to $x_2$, and $\med$ to $x_3$ respectively. Then we define an alignment between $x_1,x_2$ and $x_1,x_3$ via $\med$. We use $M_{1,\m,2}$ and $M_{1,\m,3}$ to denote $M_{\m,2}\circ M_{1,\m}$ and $M_{\m,3} \circ M_{1,\m}$ respectively. Next, we compare the alignments $A_{1,\bar{x},2}^p$ ($A_{1,\bar{x},3}^p$) and $M_{1,\m,2}$ ($M_{1,\m,3}$) with $A_{1,2}^q$ ($A_{1,3}^q$). For that purpose, we use the notations $\I$ and $\tI_{\epsilon}$ with respect to the alignments $A^p, A^q$ and $M$ (as defined in Section~\ref{sec:unique-edit}). For any $i \in [n]$, let $W^i$ denotes the interval $[i-\epsilon/q , i+\epsilon/q]$.

Recall, for any $k\in \{2,3\}$, $\tI_{M_{1,\m,k}}^{A_{1,k}^q}$ is the set of all the indices $i$ on $x_1$ such that there is exactly one edit operation inside the interval $W^i$ on $x_1$ with respect to both the alignments $M_{1,\m,k}$ and $A_{1,k}^q$. Now consider any such interval $W^i$. Since the alignment $M_{1,\m,k}$ is obtained by concatenating two alignments $M_{1,\m}$ and $M_{\m,k}$, that one edit operation inside the interval $W^i$ happens either in the alignment $M_{1,\m}$ or $M_{\m,k}$, but not in both. Now, let us consider the indices $i$ such that in the interval $W^i$ that one edit operation happens with respect to the alignment $M_{1,\m}$. This leads us to the following definition,
$$\J_{M_{1,\m,k}}^{A_{1,k}^q}\eqdef \{i \in \tI_{M_{1,\m,k}}^{A_{1,k}^q} \mid \cost_{M_{1,\m}}(W^i)=\cost_{M_{1,\m,k}}(W^i)=1\}.$$

Since in the above definition we consider the cost of the interval $W^i$ to be exactly one with respect to the alignment $M_{1,\m}$, if $M_{1,\m}(i-\epsilon/q)=\bot$ then $M_{1,\m}(i-\epsilon/q+1)\ne\bot$. Similarly, if $M_{1,\m}(i+\epsilon/q)=\bot$ then $M_{1,\m}(i+\epsilon/q-1)\ne\bot$. In words, if any one of the boundary symbols of the block $x_1[W^i]$ gets deleted then the next symbol inside the block must be mapped to some symbol in $\med$. 

Without loss of generality, from now on we assume that for all $i \in \J_{M_{1,\m,k}}^{A_{1,k}^q}$, $M_{1,\m}(i-\epsilon/q) \ne \bot$ and $M_{1,\m}(i+\epsilon/q) \ne \bot$. (If for some $i$, the above assumption is not true then we need to argue with $M_{1,\m}(i-\epsilon/q+1)$ instead of $M_{1,\m}(i-\epsilon/q)$, or with $M_{1,\m}(i+\epsilon/q-1)$ instead of $M_{1,\m}(i+\epsilon/q-1)$.)

In a similar way, define $\J_{A_{1,\bar{x},k}^p}^{A_{1,k}^q}$ as 

$$\J_{A_{1,\bar{x},k}^p}^{A_{1,k}^q}\eqdef \{i \in \tI_{A_{1,\bar{x},k}^p}^{A_{1,k}^q} \mid \cost_{A_{1,\bar{x}}^p}(W^i)=\cost_{A_{1,\bar{x},k}^p}(W^i)=1\}.$$

Now from the above definition, it is easy to observe the following.
\begin{proposition}
\label{prop:block-preserving}
For any $k \in \{2,3\}$, for all $i \in \J_{M_{1,\m,k}}^{A_{1,k}^q} \bigcap \J_{A_{1,\bar{x},k}^p}^{A_{1,k}^q}$,
$$\med[M_{1,\m}(i-\epsilon/q),M_{1,\m}(i+\epsilon/q)] = \bar{x}[A_{1,\bar{x}}^p(i-\epsilon/q),A_{1,\bar{x}}^p(i+\epsilon/q)].$$
\end{proposition}

Let us now define a set of good events and then from now on we proceed by assuming those good events occur.
\begin{itemize} \compactify
\item $\mathcal{G}_1 \eqdef$ For each of the pairs of strings $(x_1,\bar{x}), (\bar{x},x_2)$ and $(\bar{x},x_3)$, their edit distances are at least $(1-6\epsilon)p n -  2\sqrt{n} \log n$ and at most $(1+\epsilon)p n + 2\sqrt{n} \log n$.
\item $\mathcal{G}_2 \eqdef$ For any $i \ne j \in [3]$, $(1-6\epsilon)q n - 2\sqrt{n} \log n \le \ED(x_i,x_j) \le (1+\epsilon) q n + 4\sqrt{n} \log n$.
\item $\mathcal{G}_3 \eqdef$ For all $k \in \{2,3\}$, $\Big|\tI_{A_{1,\bar{x},k}^p}^{A_{1,k}^q}\Big| \ge (1-25\epsilon) q n - 2 \sqrt{n} \log n$, and $\Big|\tI_{M_{1,\m,k}}^{A_{1,k}^q}\Big| \ge (1-37\epsilon) q n - 2 \sqrt{n} \log n$.
\item $\mathcal{G}_4 \eqdef$ For any $i\ne j \in \{2,3\}$, $\Big|\J_{A_{1,\bar{x},i}^p}^{A_{1,i}^q} \setminus \J_{A_{1,\bar{x},j}}^{A_{1,j}^q} \Big| \le 2 \epsilon \Big|\J_{A_{1,\bar{x},i}^p}^{A_{1,i}^q}\Big| $.
\end{itemize}
Let $\mathcal{G}\eqdef \bigcap_{i \in [4]}\mathcal{G}_i$.
\begin{claim}
\label{clm:good-event}
The probability that the event $\mathcal{G}$ occurs is at least $1-e^{-\log^2 n}$.
\end{claim}
We defer the proof of the above claim to the end of this section. Let us now make a simple observation on the value of $\opt$.
\begin{claim}
\label{clm:opt-value}
Assuming $\mathcal{G}_1$ and $\mathcal{G}_2$ occur, $3(1-7 \epsilon)p n \le \opt \le 3(1+\epsilon) p n +  6 \sqrt{n} \log n$.
\end{claim}
\begin{proof}
Since $\bar{x}=G_p(x_1)$, $\ED(x_1,\bar{x}) \le (1+\epsilon)p n + 2 \sqrt{n} \log n$, which implies $|\I^{A_{1,\bar{x}}^p}| \le (1+\epsilon) p n + 2 \sqrt{n} \log n$. Similar bounds also hold for $\ED(x_2,\bar{x})$ and $ \ED(x_3,\bar{x})$. Thus,
\begin{align*}
\obj(\bar{x})& =\ED(x_1,\bar{x}) + \ED(x_2,\bar{x}) + \ED(x_3,\bar{x})\\
&\le 3(1+\epsilon)p n + 6 \sqrt{n} \log n &\text{assuming }\mathcal{G}_1.
\end{align*}
Clearly, $\opt \le \obj(\bar{x})$.

Now for the lower bound, let $\med^*$ be an optimal median. Thus we deduce that 
\begin{align*}
\opt &= \sum_{i=1}^3 \ED(x_i,\med^*)\\
&=\frac{1}{2} \sum_{i < j}(\ED(x_i,\med^*) + \ED(x_j, \med^*))\\
&\ge \frac{1}{2} \sum_{i < j} \ED(x_i,x_j) &\text{by triangular inequality}\\
&\ge \frac{3}{2} ((1-6\epsilon)q n - 2 \sqrt{n} \log n) &\text{assuming }\mathcal{G}_2\\
&\ge 3(1-7 \epsilon)p n &\text{since } q = 2p - \Theta(p^2) \text{ and }p \in (0,0.001).
\end{align*}
\end{proof}

Next, we provide an upper bound on the cost of an optimal alignment between $\med$ and $x_i$ for $i \in [3]$.
\begin{claim}
\label{clm:med-str-dist}
Assuming $\mathcal{G}_1$ and $\mathcal{G}_2$ occur, for each $i\in [3]$, $\ED(x_i,\med) \le (1+22\epsilon)p n$.
\end{claim}
\begin{proof}
Since $\med$ is an $(1+\epsilon)$-approximate median of $x_1,x_2,x_3$, $\obj(\med) \le (1+\epsilon) \opt$. Now, we show that $\ED(x_1,\med) \le (1+22\epsilon)p n$. For the contradiction sake, let us assume that $\ED(x_1,\med) \ge (1+22\epsilon)p n$. Then 
\begin{align*}
\ED(x_2,\med) + \ED(x_3,\med) & = \obj(\med) - \ED(x_1,\med)\\
&\le (1+\epsilon)\opt - \ED(x_1,\med)\\
&\le 3(1+\epsilon)^2 p n + 6 (1+\epsilon) \sqrt{n} \log n - (1+22\epsilon)p n &\text{by Claim~\ref{clm:opt-value}}\\
&\le (1-7\epsilon)q n &\text{since} q = 2p - \Theta(p^2).
\end{align*}
From the above we can deduce that by triangular inequality, $\ED(x_2,x_3) \le \ED(x_2,\med) + \ED(x_3,\med) \le (1-7\epsilon)q n$, which contradicts the fact that the event $\mathcal{G}_2$ occurs. A similar argument works for $\ED(x_2,\med)$ and $\ED(x_3,\med)$.
\end{proof}

The next claim is the key to prove Theorem~\ref{thm:median-trace-close-alternate}.
\begin{lemma}
\label{lem:large-similarity}
Assuming $\mathcal{G}$ occurs, for any $k \in \{2,3\}$, $\Big|\J_{M_{1,\m,k}}^{A_{1,k}^q} \bigcap \J_{A_{1,\bar{x},k}^p}^{A_{1,k}^q}\Big| \ge (1-276\epsilon)p n$.
\end{lemma} 
Before proving the above claim, let us fist prove Theorem~\ref{thm:median-trace-close-alternate} by assuming the above claim.

\begin{proof}[Proof of Theorem~\ref{thm:median-trace-close-alternate}]
Let us first assume that the good event $\mathcal{G}$ occurs. Let $K\eqdef \J_{M_{1,\m,2}}^{A_{1,2}^q} \bigcap \J_{A_{1,\bar{x},2}^p}^{A_{1,2}^q}$. By Claim~\ref{clm:med-str-dist}, $\ED(x_1,\med) \le (1+22\epsilon)p n$, which implies $|\I^{M_{1,\m}}| \le (1+22\epsilon)p n$. Note, $\J_{M_{1,\m,k}}^{A_{1,k}^q} \subseteq \I^{M_{1,\m}}$ and $\J_{A_{1,\bar{x},k}^p}^{A_{1,k}^q}\subseteq  \I^{A_{1,\bar{x}}^p}$. Then it follows from Lemma~\ref{lem:large-similarity} that $|\I^{M_{1,\m}} \setminus K| \le 298 \epsilon p n$. Also, assuming $\mathcal{G}_1$, $|\I^{A_{1,\bar{x}}^p} \setminus K| \le 278 \epsilon p n$.

Next observe, $\bar{x}$ can be transformed into $\med$ by the following alignment $B_{\bar{x},\med}$: First, apply the alignment function $(A_{1,\bar{x}}^p)^{-1}$ and then $M_{1,\m}$. Recall, by Proposition~\ref{prop:block-preserving}, for all $i \in \J_{M_{1,\m,k}}^{A_{1,k}^q} \cap \J_{A_{1,\bar{x},k}^p}^{A_{1,k}^q}$,
$\med[M_{1,\m}(i-\epsilon/q),M_{1,\m}(i+\epsilon/q)] = \bar{x}[A_{1,\bar{x}}^p(i-\epsilon/q),A_{1,\bar{x}}^p(i+\epsilon/q)]$. Hence, we deduce that $\ED(\bar{x},\med) \le \cost(B_{\bar{x},\med}) \le |\I^{A_{1,\bar{x}}^p} \setminus K| + |\I^{M_{1,\m}} \setminus K| \le 576 \epsilon p n$. 

Hence, by Claim~\ref{clm:opt-value}, $\cost(B_{\bar{x},\med}) \le 195\epsilon \cdot \opt$. The probability bound of the lemma follows from Claim~\ref{clm:good-event}.
\end{proof}

Now, it remains to show Lemma~\ref{lem:large-similarity}. For that purpose, we need the following observation.

\begin{claim}
\label{clm:large-IJ}
Assuming $\mathcal{G}_1, \mathcal{G}_2$ and $\mathcal{G}_3$ occur, for any $k \in \{2,3\}$, $\Big|\tI_{M_{1,\m,k}}^{A_{1,k}^q} \bigcap \J_{A_{1,\bar{x},k}^p}^{A_{1,k}^q}\Big| \ge (1-132 \epsilon)p n$.
\end{claim}
\begin{proof}
Let us consider any $k \in \{2,3\}$. Assuming the good event $\mathcal{G}_2$, $|\I^{A_{1,k}^q}| \le (1+\epsilon)q n + 2 \sqrt{n} \log n$. Further, by definition (see Section~\ref{sec:unique-edit}), $\tI_{M_{1,\m,k}}^{A_{1,k}^q},\tI_{A_{1,\bar{x},k}^p}^{A_{1,k}^q} \subseteq \I^{A_{1,k}^q}$. Then, assuming the good event $\mathcal{G}_3$, it is easy to observe that, $\Big|\tI_{M_{1,\m,k}}^{A_{1,k}^q} \bigcap \tI_{A_{1,\bar{x},k}^p}^{A_{1,k}^q}\Big| \ge (1-65\epsilon)q n$.

Next observe, $ \Big|\tI_{A_{1,\bar{x},k}^p}^{A_{1,k}^q} \setminus \J_{A_{1,\bar{x},k}^p}^{A_{1,k}^q}\Big| \le |\I^{A_{\bar{x},k}^p}|$ and thus at most $(1+\epsilon)p n + 2 \sqrt{n} \log n$ (assuming the good event $\mathcal{G}_1$). Hence,
\begin{align*}
\Big|\tI_{M_{1,\m,k}}^{A_{1,k}^q} \bigcap \J_{A_{1,\bar{x},k}^p}^{A_{1,k}^q}\Big| & \ge (1-65\epsilon)q n - (1+\epsilon)p n - 2 \sqrt{n} \log n\\
& \ge (1-132\epsilon)p n &\text{recall, }q = 2 p - \Theta(p^2).
\end{align*}
\end{proof}

Now we complete the proof of Lemma~\ref{lem:large-similarity}.
\begin{proof}[Proof of Lemma~\ref{lem:large-similarity}]
Let $R\eqdef \J_{A_{1,\bar{x},2}^p}^{A_{1,2}^q} \bigcap \J_{A_{1,\bar{x},3}^p}^{A_{1,3}^q}$ and $Q\eqdef \tI_{M_{1,\m,2}}^{A_{1,2}^q}\bigcap \tI_{M_{1,\m,3}}^{A_{1,3}^q}$.

\begin{align*}
\Big|\tI_{M_{1,\m,2}}^{A_{1,2}^q}\bigcap R \Big| &= \Big|\tI_{M_{1,\m,2}}^{A_{1,2}^q} \bigcap \J_{A_{1,\bar{x},2}^p}^{A_{1,2}^q}\Big| - \Big| \tI_{M_{1,\m,2}}^{A_{1,2}^q} \bigcap \Big( \J_{A_{1,\bar{x},2}^p}^{A_{1,2}^q} \setminus \J_{A_{1,\bar{x},3}}^{A_{1,3}^q} \Big) \Big|\\
& \ge \Big|\tI_{M_{1,\m,2}}^{A_{1,2}^q} \bigcap \J_{A_{1,\bar{x},2}^p}^{A_{1,2}^q}\Big| - \Big|\J_{A_{1,\bar{x},2}^p}^{A_{1,2}^q} \setminus \J_{A_{1,\bar{x},3}}^{A_{1,3}^q} \Big|\\
& \ge (1-132\epsilon)p n -2 \epsilon \Big|\J_{A_{1,\bar{x},2}^p}^{A_{1,2}^q}\Big| \quad \text{by Claim~\ref{clm:large-IJ} and assuming }\mathcal{G}_4\\
&\ge (1-135\epsilon) p n
\end{align*}
where the last inequality uses the fact that $\J_{A_{1,\bar{x},2}^p}^{A_{1,2}^q} \subseteq \I^{A_{1,\bar{x}}^p}$ and assuming the good event $\mathcal{G}_1$, $|\I^{A_{1,\bar{x}}^p}| \le (1+\epsilon)p n + 2 \sqrt{n} \log n$.

Similarly, $\Big|\tI_{M_{1,\m,3}}^{A_{1,3}^q}\cap R \Big| \ge (1-135\epsilon) p n$. Since $\J_{A_{1,\bar{x},2}^p}^{A_{1,2}^q},\J_{A_{1,\bar{x},3}^p}^{A_{1,3}^q} \subseteq \I^{A_{1,\bar{x}}^p}$, $|R|\le (1+\epsilon) p n + 2 \sqrt{n} \log n$. Hence, we can deduce that, 
\begin{equation}
\label{eq:Q-R-bound}
 |Q\cap R| \ge (1-272 \epsilon)p n   
\end{equation}

Next, we claim that $\Big| (Q \cap R) \setminus \J_{M_{1,\m,2}}^{A_{1,2}^q} \Big| \le 4 \epsilon p n$. To prove the claim, observe, for each $i \in R$, 
$$\bar{x}[A_{1,\bar{x}}^p(i-\epsilon/q), A_{1,\bar{x}}^p(i+\epsilon/q)] = x_2[A_{1,2}^q(i-\epsilon/q), A_{1,2}^q(i+\epsilon/q)] = x_3[A_{1,3}^q(i-\epsilon/q), A_{1,3}^q(i+\epsilon/q)].$$

Then consider any $i \in (Q \cap R) \setminus \J_{M_{1,\m,2}}^{A_{1,2}^q}$. 
Observe, the followings hold directly from the definition.
\begin{enumerate}
\item $x_1[i-\epsilon/q , i+\epsilon/q] = \med[M_{1,\m}(i-\epsilon/q) , M_{1,\m}(i+\epsilon/q)]$.
\item For any $k \in \{2,3\}$, the cost of the (sub-)alignment $M_{\m,k}$ restricted to the mapping from the substring $\med[M_{1,\m}(i-\epsilon/q) , M_{1,\m}(i+\epsilon/q)]$ to $x_k[A_{1,2}^q(i-\epsilon/q), A_{1,2}^q(i+\epsilon/q)]$ is exactly one.
\end{enumerate}
Thus we can modify the string $\med$ as follows: Take each $i \in (Q \cap R) \setminus \J_{M_{1,\m,2}}^{A_{1,2}^q}$, and replace the block $\med[M_{1,\m}(i-\epsilon/q) , M_{1,\m}(i+\epsilon/q)]$ by $\bar{x}[A_{1,\bar{x}}^p(i-\epsilon/q), A_{1,\bar{x}}^p(i+\epsilon/q)]$.
 We call the resulting string $\med'$. Note, for any $k\in \{2,3\}$, $\ED(\med',x_k) \le \ED(\med,x_k) - \Big| (Q \cap R) \setminus \J_{M_{1,\m,2}}^{A_{1,2}^q} \Big|$. On the other hand, $\ED(\med',x_1) \le \ED(\med,x_1) + \Big| (Q \cap R) \setminus \J_{M_{1,\m,2}}^{A_{1,2}^q} \Big|$. Hence, we get that
 $$\obj(\med') \le \obj(\med) - \Big| (Q \cap R) \setminus \J_{M_{1,\m,2}}^{A_{1,2}^q} \Big|.$$
 Now, if $\Big| (Q \cap R) \setminus \J_{M_{1,\m,2}}^{A_{1,2}^q} \Big| > 4\epsilon p n$, then 
 \begin{align*}
 \obj(\med') &\le \obj(\med) - 4\epsilon p n\\
 & \le (1+\epsilon)\opt - 4 \epsilon p  n &\text{since }\med \text{ is an }(1+\epsilon)\text{-approximate median}\\
 & < (1+\epsilon)\opt - \epsilon \opt &\text{by Claim~\ref{clm:opt-value}}\\
 & = \opt
 \end{align*}
which leads us to a contradiction since $\opt \le \obj(\med')$. So we conclude that $\Big| (Q \cap R) \setminus \J_{M_{1,\m,2}}^{A_{1,2}^q} \Big| \le 4 \epsilon p n$. So we can deduce the following
\begin{align*}
\Big|\J_{M_{1,\m,2}}^{A_{1,2}^q} \cap \J_{A_{1,\bar{x},2}^p}^{A_{1,2}^q}\Big| & \ge \Big| J_{M_{1,\m,2}}^{A_{1,2}^q} \bigcap (Q \cap R)\Big|\\
&= |Q\cap R| - \Big| (Q \cap R) \setminus \J_{M_{1,\m,2}}^{A_{1,2}^q} \Big|\\
&\ge (1-276\epsilon) p n&\text{by~\eqref{eq:Q-R-bound}}.
\end{align*}
Similarly, we can also show that $\Big|\J_{M_{1,\m,3}}^{A_{1,3}^q} \cap \J_{A_{1,\bar{x},3}^p}^{A_{1,3}^q}\Big| \ge (1-276\epsilon)p n$.
\end{proof}

\begin{proof}[Proof of Claim~\ref{clm:good-event}]
First, we have already observed that $x_2,x_3 \sim G_q(x_1)$ (when viewed them as infinite length strings). Further, by Proposition~\ref{prop:dist-symmetry} Proposition~\ref{prop:dist-transitivity}, we can say that $x_3 \sim G_q(x_2)$. We use the notation $A_{x_2,x_3}^q$ to denote the alignment produced by the random process $G_q(x_2)$. Then the followings are immediate from the Chernoff-Hoeffding bound. For any $i \in [n]$,
\begin{itemize}
\item $\Pr[|A_{1,2}^q (i) - i| > \sqrt{n} \log n] \le 2e^{-2\log^2 n}$,
\item $\Pr[|A_{1,3}^q (i) - i| > \sqrt{n} \log n] \le 2e^{-2\log^2 n}$,
\item $\Pr[|A_{x_2,x_3}^q (i) - i| > \sqrt{n} \log n] \le 2e^{-2\log^2 n}$,
\item $\Pr[|A_{1,\bar{x}}^p (i) - i| > \sqrt{n} \log n] \le 2e^{-2\log^2 n}$,
\item $\Pr[|A_{\bar{x},3}^p (i) - i| > \sqrt{n} \log n] \le 2e^{-2\log^2 n}$,
\item $\Pr[|A_{\bar{x},2}^p (i) - i| > \sqrt{n} \log n] \le 2e^{-2\log^2 n}$.
\end{itemize}
So, by a union bound, we get that for all $i \in [n]$, all the above six conditions hold with probability at least $1 - e^{-1.5\log^2 n}$ (for large enough $n$). Below while considering the $n$-length prefixes, we assume that the above event, denoted by $\mathcal{G}_5$, holds.

Recall, $\bar{x}=G_p(x_1)$  (when viewed them as infinite length strings). Then by Lemma~\ref{lem:number-edit-operation} together with the assumption of $\mathcal{G}_5$, for any $\epsilon > 0$, with probability at least $1-2e^{-\epsilon^2pn/3}$, $\ED(x_,\bar{x}) \le (1+\epsilon)p n + 2\sqrt{n} \log n$. Further, by Lemma~\ref{lem:editbound-small-alphabet}, we get that with probability at least $1-2e^{-\epsilon^2 p^2 n/2}$, $\ED(x_1,\bar{x}) \ge (1-6\epsilon)p n - 2\sqrt{n} \log n$. Similar bounds holds for the pairs $(\bar{x},x_2), (\bar{x},x_3)$ and $(x_i,x_j)$ for any $i\ne j \in [3]$.

Next, we assume $\mathcal{G}_1$, $\mathcal{G}_2$ and $\mathcal{G}_5$ occur. Then given that, we want to claim that $\mathcal{G}_3$ occurs with probability at least $1-12e^{-\epsilon^2 q^2 n/5}$. Note, by our choice of $\epsilon$, it holds that $\epsilon \ge 54 q\log_2 \frac{1}{q}$. Note, by triangular inequality, $\cost(A_{1,\bar{x},2}^p) \le \cost(A_{1,\bar{x}}^p) + \cost(A_{\bar{x},2}^p)$. Then by Lemma~\ref{lem:number-edit-operation}, $\cost(A_{1,\bar{x},2}^p) \le (1+\epsilon)2p n $ with probability at least $1-4 e^{-\epsilon^2p n/4}$. Note, $q = 2p - \Theta(p^2)$. Then by Lemma~\ref{lem:unique-alignment}, $\Big|\tI_{A_{1,\bar{x},2}^p}^{A_{1,2}^q}\Big| \ge (1-25\epsilon) q n - 2 \sqrt{n} \log n$ with probability at least $1-4e^{-\epsilon^2 q^2 n/2}$. A similar argument holds for $\Big|\tI_{A_{1,\bar{x},3}^p}^{A_{1,3}^q}\Big|$. 

Next, we argue about the set $\tI_{M_{1,\m,2}}^{A_{1,2}^q}$. By Claim~\ref{clm:med-str-dist}, $\ED(\med,x_i) \le (1+22\epsilon)p n$ for each $i\in [3]$. Then, by triangular inequality, $\cost(M_{1,\m,2}) \le 2(1+22\epsilon)p n \le (1+23\epsilon)q n$. Then by Lemma~\ref{lem:editbound-small-alphabet} (and assuming $\mathcal{G}_5$ occurs), with probability at least $1-2e^{-(23\epsilon)^2 q^2 n/2}$, $\Big|\tI_{M_{1,\m,2}}^{A_{1,2}^q}\Big| \ge (1-37\epsilon) q n - 2 \sqrt{n} \log n$. A similar argument holds for $\Big|\tI_{M_{1,\m,3}}^{A_{1,3}^q}\Big|$.

Next, we assume $\mathcal{G}_1$, $\mathcal{G}_2$, $\mathcal{G}_3$ and $\mathcal{G}_5$ occur. For $\mathcal{G}_4$, observe, $\mathbb{E}\Big[ |\J_{A_{1,\bar{x},i}^p}^{A_{1,i}^q} \setminus \J_{A_{1,\bar{x},j}}^{A_{1,j}^q} \Big| \Big] = \epsilon \Big|\J_{A_{1,\bar{x},i}^p}^{A_{1,i}^q}\Big|$. Then by a standard application of Chernoff bound, we get that the probability that $\Big|\J_{A_{1,\bar{x},i}^p}^{A_{1,i}^q} \setminus \J_{A_{1,\bar{x},j}}^{A_{1,j}^q} \Big| > 2 \epsilon \Big|\J_{A_{1,\bar{x},i}^p}^{A_{1,i}^q}\Big|$ is at most $2^{-\epsilon \Big|\J_{A_{1,\bar{x},i}^p}^{A_{1,i}^q}\Big|/2} \le 2^{-\frac{\epsilon p n}{4}}$ (by Claim~\ref{clm:large-IJ}).

 So finally we get that $\mathcal{G}$ holds (i.e., all the $\mathcal{G}_i$, for $i \in [4]$ hold) with probability at least $1-e^{-\log^2 n}$.
\end{proof}

\paragraph{For more than three traces.}
So far, we have shown that for any set $\{s_1,s_2,s_3\}$ of three traces of $s$, its any $(1+\epsilon)$-approximate median is close to $s$. Below we argue that a similar result for any arbitrary number (less than some $\mathrm{poly}(n)$) of traces directly follows.

\begin{corollary}
\label{cor:m-traces-median}
For a large enough $n \in \N$ and a noise parameter $p \in (0,0.001)$, 
let the string $s\in\{0,1\}^n$ be chosen uniformly at random, 
and let $s_1,\cdots,s_m$ be $m=n^{O(1)}$ traces generated by $G_p(s)$. 
If $\med$ is a $(1+\epsilon)$-approximate median of $S=\{s_1,\cdots,s_m\}$
for any $\epsilon \in [110 p \log (1/p),1/6]$, 
then $\Pr [\ED(s,\med) \le O(\epsilon)\cdot \frac{\opt(S)}{m} ] \geq 1-n^{-1}$.
\end{corollary}
\begin{proof}[Proof (Sketch).] 
By an argument similar to that used in Claim~\ref{clm:opt-value}, we can show that with high probability $\opt(S) \leq \obj(S,s) \le (1+\epsilon) pnm$. Let $z$ be an $(1+\epsilon)$-approximate median of $S$,
then $\obj(S,z) \le (1+\epsilon)^2 pnm$.
By averaging, there exists a subset $S'=\{s_i,s_{i+1},s_{i+2}\} \subseteq S$ such that $\obj(S',z) \le 3(1+\epsilon)^2 pn$, hence $z$ is a $(1+O(\epsilon))$-approximate median of $S'$. Thus by Theorem~\ref{thm:median-trace-close-alternate} together with a union bound, $z$ is at distance at most $O(\epsilon) \opt(S')$ from $s$. 
We can also show that with high probability $\opt(S') \ge (1 - 7\epsilon)pnm$,
again by an argument similar to Claim~\ref{clm:opt-value}. 
We thus conclude that with high probability
$ \ED(s,z) \le O(\epsilon) \frac{\opt(S)}{m}$.
\end{proof}

\section{Linear-time Approximate Trace Reconstruction}
\label{sec:trace-median-algo}
In this section, we describe a linear-time algorithm that reconstructs the unknown string using only three traces, up to some small edit error. In particular, we prove Theorem~\ref{thm:main}. Before describing our linear-time algorithm, first note that we can compute an (exact) median of three traces using a standard dynamic programming algorithm~\cite{Sankoff75, kruskal1983} in cubic time. Then by Theorem~\ref{thm:median-trace-close-alternate}, that median string will be close (in edit distance) to the unknown string. More specifically, the edit distance between the computed median string and the unknown string will be at most $O(\epsilon p n)$ (follows from Theorem~\ref{thm:median-trace-close-alternate} and Claim~\ref{clm:opt-value}) with high probability. In this section, we design a more sophisticated method to compute an approximation of the unknown string. For that purpose, we first divide each trace into "well-separated" blocks of size $\log^2 n$ each. Then we run the dynamic-programming-based median algorithm~\cite{Sankoff75, kruskal1983} on these small blocks. Thus we spend only $\mathrm{poly} \log n$ time per block, and hence in total $\tilde{O}(n)$ time. Since we consider "well-separated" blocks, they are independent. Thus we apply Theorem~\ref{thm:median-trace-close-alternate} for each of these blocks (instead of the whole string). Using standard Chernoff-Hoeffding bound, we get that most of these block medians are close to their corresponding block of the unknown string. Hence, by concatenating these block medians, we get back the whole unknown string up to some small edit error. Formally, our result is the following.

\begin{theorem}
\label{thm:linear-trace}
There is a small non-negative constant $c_0 < 1$ and a deterministic algorithm that, 
given as input a noise parameter $p \in (0,c_0]$, 
an accuracy parameter $\epsilon \in [110 p \log (1/p), 1/6]$, 
and three traces $s_1,s_2,s_3 \sim G_p(s)$ for a uniformly random (but unknown) string $s\in \{0,1\}^n$, 
outputs in time $\tilde{O}(n)$ a string $z$ that satisfies 
$\Pr [\ED(s,z) \le  5270 \epsilon p n ] \geq 1-n^{-1}$.
\end{theorem}

Before describing the algorithm we would like to introduce some notation that we use in this section. 
For a string $x \in \Sigma^n$, let $y=x[i,i+1,\cdots,j]$ be a substring of it; 
then we denote by $\texttt{start}(y)$ the index $i$ and by $\texttt{end}(y)$ the index $j$. 

\paragraph{Description of the algorithm. }
The algorithm works as follows. First, partition $s_1$ into $r=\frac{|s_1|}{\ell}$ disjoint blocks $s_1^{1},s_1^{2},\cdots,s_1^{r}$ each of length $\ell=\log^2 n + \frac{240}{p} \log^{3/2} n$. For each $s_1^i$, let us call the middle $\log^2 n$-size sub-block, denoted by $y_1^i$, an \emph{anchor}. Next, for each $i \in [r]$ and $j \in \{2,3\}$, find the \emph{best match} of the anchor $y_1^i$ in the string $s_j$ i.e., for each $y_1^i$ find a substring (breaking ties arbitrarily) in $s_j$ that has the minimum edit distance with $y_1^i$. Let us denote the matched substrings in $s_2$ and $s_3$ by $y_2^i$ and $y_3^i$ respectively. Then for each $j \in \{2,3\}$, we divide $s_j$ into blocks $s_j^{1},\cdots,s_j^{r}$ (some of the $s_j^i$'s could be empty) by treating $y_j^{1},\cdots,y_j^{r}$ as anchors. More specifically,
\begin{itemize}
\item Set the start index of $s_j^1$ to be 1. For any other non-empty block $s_j^i$, set its start index to be $\lfloor(\texttt{end}(y_j^{i-1}) + \texttt{start}(y_j^i))/2 \rfloor$.
\item For the last non-empty block in $s_j$, set its end index to be $|s_j|$. For any other non-empty block $s_j^i$, set its end index to be $\lfloor(\texttt{end}(y_j^{i}) + \texttt{start}(y_j^{i+1}))/2 \rfloor - 1$.
\end{itemize}
Next, for each $i \in [r]$, compute a median of $\{s_1^i,s_2^i,s_3^i\}$, and let it be denoted by $z^i$. Finally, output $z = z^1 \odot \cdots \odot z^r$ (i.e., the concatenation of all the $z^i$'s).

\paragraph{Correctness proof. }
Before proceeding with the correctness proof, let us state a known fact from~\cite{BEKMRRS03} about the edit distance between two random strings.
\begin{proposition}[\cite{BEKMRRS03}]
\label{prop:random-string-edit}
For any two strings $x\in \Sigma^m$ and $y \in \Sigma^n$ drawn uniformly at random, $\Pr[\ED(x,y) \ge \frac{\max\{m,n\}}{10}] \ge 1-2^{-\max\{m,n\}/10}$.
\end{proposition}

\begin{claim}
\label{clm:inter-block-dist}
For every two substrings $x,y$ of length at least $60 \log n$, of $s_i,s_j$ respectively, where $i,j \in [3]$, such that $(A_{s,s_i}^p)^{-1}(x)$ and $(A_{s,s_j}^p)^{-1}(y)$ are two disjoint substrings of $s$, $\Pr[\ED(x,y) \ge \frac{\max\{|x|,|y|\}}{10}] \ge 1-n^{-4}$.
\end{claim}
\begin{proof}
By Proposition~\ref{prop:dist-symmetry}, $x$ and $y$ are two strings chosen uniformly at random by picking each of its symbols independently uniformly at random from $\Sigma$. Then the claim directly follows from Proposition~\ref{prop:random-string-edit} together with a standard application of union bound.
\end{proof}

Let us now define a \emph{true match} for each block $y_1^i$ in strings $s_2$ and $s_3$. For each $j \in \{2,3\}$, we call the block $A_{s,s_j}^p((A_{s,s_1}^p)^{-1}(y_1^i))$ in $s_j$ the true match of $y_1^i$, denoted by $t_j^i$. Next, we want to claim that for each block $y_1^i$, its best match $y_j^i$ in a string $s_j$, for $j\in \{2,3\}$, is close to its true match $t_j^i$. The following lemma is crucial to show the correctness of the algorithm and also to establish a linear-time bound for the algorithm.

\begin{claim}
\label{clm:near-perfect-align}
For each $i \in [r]$ and $j \in \{2,3\}$, with probability $1-5n^{-2}$,
\begin{enumerate}
\item $|\texttt{start}(t_j^i) - \texttt{start}(y_j^i)| \le \frac{200}{p} \log n$, and
\item $|\texttt{end}(t_j^i) - \texttt{end}(y_j^i)| \le \frac{200}{p} \log n$.
\end{enumerate}
\end{claim}
\begin{proof}
Let us partition $y_1^i$ into $\frac{p \log n}{10}$ sub-blocks $y_1^{i,1},\cdots,y_1^{i,(p \log n)/10}$, each of size $\frac{10\log n}{p}$. Next, for each of these sub-blocks $y_1^{i,k}$ consider its true match in the string $s_j$ (for any $j \in \{2,3\}$) defined as $t_j^{i,k}\eqdef A_{s,s_j}^p((A_{s,s_1}^p)^{-1}(y_1^{i,k}))$.

Now, consider an (arbitrary) optimal alignment $B$ between $y_1^i$ and $y_j^i$. Observe, if for all $1\le k \le (p \log n)/10$, $B(y_1^{i,k})$ has a non-empty overlap with the corresponding true match $t_j^{i,k}$, then the claim is true. So from now on, let us assume that at least for some block $y_1^{i,k}$, $B(y_1^{i,k})$ does not overlap with $t_j^{i,k}$. Let $y_1^{i,k'}$ be the right-most sub-block such that there is a non-empty overlapping between $B(y_1^{i,k'})$ and $t_j^{i,k'}$. Then for each $k'+1 \le k \le (p \log n)/10$, by Claim~\ref{clm:inter-block-dist}, $\cost_B(y_1^{i,k}) \ge \frac{\log n}{p}$ with probability at least $1-n^{-4}$.

We want to claim that $k' \ge \frac{p \log n}{10} - \frac{10}{1-24p}$. If not, then we deduce that $y_j^i$ is not the best match of $y_1^i$ in the string $s_j$. To argue this, suppose $k' < \frac{p \log n}{10} - \frac{10}{1-24p}$. Then we modify the mapping $B$ to derive another mapping $B'$ as follows: $B'$ respects $B$ till the block $y_1^{i,k'-1}$. Next, $B'$ deletes the block $y_1^{i,k'}$, and then use the mapping $A_{s,s_j}^p \circ (A_{s,s_1}^p)^{-1}$ to map the remaining blocks $y_1^{i,k'+1},\cdots,y_1^{i,(p \log n)/10}$. By Lemma~\ref{lem:number-edit-operation} (applied on strings of size at least $\frac{10 \log n}{p}$) together with an union bound, we get that for all the blocks of $s_1$ of size at least $\frac{10 \log n}{p}$, the cost of the alignment $A_{s,s_j}^p \circ (A_{s,s_1}^p)^{-1}$ is at most $24\log n$ with probability at least $1-2n^{-2}$.

Clearly, the cost of this new alignment $B'$ is at least $(\frac{10}{1-24p} + 1)\frac{\log n}{p} - (\frac{10 \log n}{p} + \frac{240 \log n}{1-24p} ) > 0$ less than that of $B$. Hence, $y_j^i$ cannot be the best match of $y_1^i$ in $s_j$. So we deduce that $k' \ge \frac{p \log n}{10} - \frac{10}{1-24p}$. Note, if an alignment function just deletes all the blocks $y_1^{i,k'+1},\cdots,y_1^{i,(p \log n)/10}$, it would cost at most $\frac{100 \log n}{p(1-24p)}$. Thus, since $t_j^i$ is the best match of $y_1^i$, the cost of $B$ for these blocks $y_1^{i,k'+1},\cdots,y_1^{i,(p \log n)/10}$ must be at most $\frac{100 \log n}{p(1-24p)} $. From this we conclude that $|\texttt{end}(t_j^i) - \texttt{end}(y_j^i)| \le \frac{100}{p(1-24p)} \log n \le \frac{200}{p} \log n$ (for the choice of $p$ we have).

Similarly, we can argue that $|\texttt{start}(t_j^i) - \texttt{start}(y_j^i)| \le \frac{200}{p} \log n$. This concludes the proof.
\end{proof}

The following is an immediate corollary of the above claim.
\begin{corollary}
\label{cor:non-overlap-match}
With probability at least $1-10n^{-2}$, for each $i\in [r]$ and $j \in \{2,3\}$, $y_j^i$ and $y_j^{i+1}$ do not overlap.
\end{corollary}
\begin{proof}
Consider the substring between the blocks $y_1^i$ and $y_1^{i+1}$, which is of length $\frac{480}{p} \log^{3/2} n$. By Lemma~\ref{lem:number-edit-operation}, $A_{s,s_j}^p\circ (A_{s,s_1}^p)^{-1}$ maps that substring into a substring of length at least $240\log^{3/2} n > \frac{400}{p} \log n$ in $s_j$ with probability at least $1-n^{-3}$. Now, it directly follows from Claim~\ref{clm:near-perfect-align} that $y_j^i$ and $y_j^{i+1}$ do not overlap.
\end{proof}

Next, we use the above to establish an upper bound on the edit distance between the unknown string $s$ and the recovered string $z$.

\begin{lemma}
\label{lem:correctness}
With probability at least $1-n^{-1}$, $\ED(s,z) \le 1550 \epsilon p n$.
\end{lemma}
\begin{proof}
For any $i \in [r]$, consider the set $S^i \eqdef \{s_1^i,s_2^i,s_3^i\}$. Consider the substring $y^i$ of the string $s$ such that $A_{s,s_1}^p(y^i)=y_1^i$ (i.e., $y^i$ maps to $y_1^i$ by the alignment $A_{s,s_1}^p$). Next, for the analysis purpose, consider the set $T^i = \{y_1^i, t_2^i , t_3^i\}$. Recall, by the definition of $t_j^i = A_{s,s_j}^p(y^i)$, for $j \in \{2,3\}$ are the traces generated by $G_p$ from the block $y^i$.

Thus by Lemma~\ref{lem:number-edit-operation}, with probability at least $1-3n^{-4}$, for each $t \in T^i$, $\ED(y^i,t) \le (1+\epsilon)p |y^i|$. Since $y_1^i$ is a substring of $s_1^i$ (where $|s_1^i| = |y_1^i| + \frac{240}{p} \log^{3/2} n$), by triangular inequality, $\ED(y^i , s_1^i) \le (1+\epsilon)p |y^i| + \frac{240}{p} \log^{3/2} n$. Next observe, for each $j \in \{2,3\}$, by Corollary~\ref{cor:non-overlap-match}, $t_j^i$ is a substring of $s_j^i$. Furthermore, by definition, $|s_j^i| \le  |y_j^i| + \frac{240}{p} \log^{3/2} n$. Thus, again by triangular inequality, $\ED(y^i , s_j^i) \le (1+\epsilon)p |y^i| + \frac{240}{p} \log^{3/2} n$. So we get 
\begin{align}
\label{eq:obj-bound}
\obj(S^i, y^i) \le 3 (1+\epsilon)p |y^i| +  \frac{750}{p} \log^{3/2} n.
\end{align}
Since $z^i$ is an (exact) median of $S^i$, $\obj(S^i , z^i) \le \obj(S^i, y^i)$. Next, it follows from Claim~\ref{clm:near-perfect-align} and the construction of the blocks $s_j^i$, for $j \in \{2,3\}$, that $t_j^i$ is a substring of $s_j^i$ where $|t_j^i| \ge |s_j^i| - \frac{500}{p} \log^{3/2} n$. Hence, we can deduce that
\begin{align}
\label{eq:obj}
\obj(T^i, z^i) &\le \obj(S^i, z^i) + \frac{1500}{p} \log^{3/2} n \nonumber\\
&\le \obj(S^i, y^i) + \frac{1500}{p} \log^{3/2} n \nonumber\\
&\le 3 (1+\epsilon)p |y^i| +  \frac{2500}{p} \log^{3/2} n&\text{by~\eqref{eq:obj-bound}}.
\end{align}

Further observe, it follows from Proposition~\ref{prop:dist-transitivity} and Lemma~\ref{lem:editbound-small-alphabet}, for each $j \in \{2,3\}$,
$$\ED(y_1^i,t_j^i) \ge (1-6\epsilon)q |y_1^i|=(1-6\epsilon)q \log ^2 n.$$
Recall, $q=2p - \Theta(p^2)$. 
Then by an argument similar to that used in the proof of Claim~\ref{clm:opt-value}, 
\begin{align}
\label{eq:opt}
\opt(T^i) \ge 3(1-7\epsilon)p \log^2 n.
\end{align}
From~\eqref{eq:obj} and~\eqref{eq:opt}, we conclude that $z^i$ is an $(1+9\epsilon)$-approximate median of $T^i$. So by Theorem~\ref{thm:median-trace-close-alternate}, $\ED(y^i,z^i) \le 1755 \epsilon \cdot \opt(T^i)$ with probability at least $1-e^{-2 \log^2 (\log n)}$.
 
Now, since all the $y^i$'s are generated by picking each symbol uniformly at random and by our construction for each $j \in \{2,3\}$ $t_j^i$'s are disjoint, the sets $T^i$'s are independent. Hence, by applying standard Chernoff-Hoeffding bound, we get that with probability at least $1-n^{-1}$, all but at most $ e^{-2 \log^2 (\log n)} r + \sqrt{r \log r}$ many blocks satisfy, $\ED(y^i,z^i) \le 1755 \epsilon \cdot \opt(T^i)$. 

Let $s'$ denote the string $y^1\odot \cdots \odot y^r$. 
Note as $s'$ is a subsequence of $s$, 
\begin{equation*}
  \ED(s,s')=|s|-|s'|\le \frac{500 n}p{\log^{1/2} n}  
\end{equation*}

Then,
\begin{align*}
\ED(s,z) & \le \ED(s',z) + \frac{500 n}{p\log^{1/2} n}&\text{by triangular inequality}\\
&\le \sum_{i=1}^r \ED(y^i,z^i) + \frac{500 n}{p\log^{1/2} n}\\
&\le 1755 \epsilon \sum_{i=1}^r \opt(T^i) + (e^{-2 \log^2 (\log n)} r + \sqrt{r \log r}) 2 \log^2 n + \frac{500 n}{p\log^{1/2} n}\\
&\le 5270 \epsilon p n &\text{by~\eqref{eq:obj} and }r=\Theta(n/\log^2 n).
\end{align*}
\end{proof}

\paragraph{Running time analysis. }
Partitioning the string $s_1$ into $r$ blocks clearly takes linear time. The main challenge here is to find the best match $y_j^i$ (for $j \in \{2,3\}$) for each block $y_1^i$. To do this, for each $j \in \{2,3\}$, we start with the first $10\log^2 n$-sized substring of $s_j$ and run the approximate pattern matching algorithm under the edit metric by~\cite{LV89, GP90} to find the best match $y_j^1$ for $y_1^1$ (which takes $O(\log^4 n)$ time). Next, we consider the $10\log^2 n$-sized substring of $s_j$ starting from the end index of $y_j^1$, and in a similar way find the best match $y_j^2$ for $y_1^2$. We continue until we find the best matches for all the blocks $y_1^1,\cdots,y_1^r$. Lemma~\ref{lem:number-edit-operation} ensures that $y_j^1$ indeed lies on the first $10\log^2 n$-sized substring of $s_j$ with probability at least $1-n^{-4}$. Then Corollary~\ref{cor:non-overlap-match} together with Lemma~\ref{lem:number-edit-operation} guarantees that to find the best match for a block $y_1^i$, it suffices to look into the $10\log^2 n$-sized substring of $s_j$ after the best match of the previous block $y_1^{i-1}$. Hence, we can identify the best matches for all the blocks $y_1^1,\cdots,y_1^r$ in time $\tilde{O}(n)$ (since $r=\Theta(\frac{n}{\log^2 n})$). Once we get $s_1^i,s_2^i,s_3^i$ for each $i\in [r]$, we can compute their median using the dynamic programming algorithm~\cite{Sankoff75, kruskal1983} in time $O(\log^6 n)$ time. So, the total running time is $\tilde{O}(n)$.

\section{Conclusion}
\label{sec:conclusion}
Trace reconstruction in the average case is a well-studied problem. The problem is to reconstruct an unknown (random) string by reading a few traces of it generated via some noise (insertion-deletion) channel. The main objective here is to minimize the sample complexity and also the efficiency of the reconstruction algorithm. There is an exponential gap between the current best upper and lower bound in the sample complexity despite several attempts. The best lower bound is $\tilde{\Omega}(\log^{5/2} n)$~\cite{C20b}. A natural question is whether it is possible to beat this lower bound by allowing some error in the reconstructed string. This version is also referred to as the approximate trace reconstruction problem; however, nothing is known except for a few special cases.

Our result not only beats the lower bound of the exact trace reconstruction but uses only three traces. The reconstructed string is $O(\epsilon p n)$ close (in edit distance) to the unknown string with high probability. We establish a connection between the approximate trace reconstruction and the approximate median string problem, another utterly significant problem. We show that both the problems are essentially the same. We leverage this connection to design a near-linear time approximate reconstruction algorithm using three traces.

An exciting future direction is to get a similar result for the worst-case, where the unknown string is arbitrary. It will also be fascinating if we could show some non-trivial sample complexity lower bound for that version.

\ifprocs
\bibliographystyle{plainurl}
\else
\bibliographystyle{alphaurl}
\fi
\bibliography{ref}


\appendix

\section{Proof of Lemma~\ref{lem:unique-alignment}}
\label{app:unique-alignment}
\begin{proof}[Proof of Lemma~\ref{lem:unique-alignment}]
We follow the outline of the proof of Lemma~\ref{lem:editbound-small-alphabet}.
Let $\mathcal{B}$ be the event that there exists an alignment $M$ between $x,y$ with  $\cost(M)\le (1+\delta)pn$ such that $|\mathcal{\tilde{I}}_M|<(1-23\epsilon-\delta)pn$. Moreover let $\mathcal{B}^0$ be the event that there exists an alignment $\bar{M}$ between $x,y$ such that for at least $6\epsilon p n$ indices $i\in \mathcal{\tilde{I}}$, $\cost_{\bar{M}}(x[i-\frac{\epsilon}{p},i+\frac{\epsilon}{p}])=0$.
In the rest of the proof we shall bound the probability of event $\mathcal{B}$. 
%
We assume henceforth that $A^p$ is known (i.e., we condition on $A^p$),
and that $|\mathcal{\tilde{I}}| > (1-5\epsilon)p n$,
which occurs with high probability by Lemma~\ref{lem:tildeI}. 
The probabilistic analysis below will use the randomness of $x$ and of the characters inserted into $y$. 

Our plan is to define some basic events $\mathcal{E}_{S,\bar{S},\Gamma}$ for every two subset $S,\bar{S}\subseteq [n]$ of the same size $\ell=|S|=|\bar{S}|$, representing positions in $x$ and in $y$, respectively, and for every vector $\gamma\in\{-1,1\}^\ell$. 
We then show that the bad events are bounded by these events and event $\mathcal{B}^0$

\begin{align} \label{eq:alignbound1}
  \mathcal{B}
  \subseteq  
  (\bigcup_{S,\bar{S},\gamma \mid \ell= \epsilon pn} \mE_{S,\bar{S},\gamma} )\bigcup \mathcal{B}^0,
\end{align}
and bound the probability of each basic event by
\begin{align} \label{eq:alignbound2}
  \Pr[ \mE_{S,\bar{S},\gamma} ] 
  \leq 
  |\Sigma|^{ -\epsilon\ell/(6p) } .
\end{align} 
The proof will then follow easily using a union bound and a simple calculation.

To define the basic event $\mathcal{E}_{S,\bar{S},\gamma}$, we need some notations. 
Write $S=\{i_1,\dots,i_\ell\}$ in increasing order, 
and similarly $\bar{S}=\{\bar{i}_1,\ldots,\bar{i}_\ell\}$,
and let $\gamma=(\gamma_1,\ldots,\gamma_\ell)$. 
Use these to define $\ell$ blocks in $x$, 
namely, $B_{i_j}=x[i_j,i_j+\frac{2\epsilon}{p}]$,
and $\ell$ blocks in $y$, namely, $\bar{B}_{i_j}=y[\bar{i}_j,\bar{i}_j+\frac{2\epsilon}{p}+\gamma_j]$. 
Note that here $i_j$ and $\bar{i}_j$ are at the beginning of their blocks 
(while in Lemma~\ref{lem:editbound-small-alphabet} they were at the middle),
and that the length of $B_{i_j}$ and of $\bar{B}_{i_j}$ might differ
(they are $1+\frac{2\epsilon}{p}$ and $\frac{2\epsilon}{p}+1+\gamma_j$, respectively).
Now define $\mE_{S,\bar{S},\gamma}$ to be the event that
(i) $\tilde{S}\eqdef \{i_1+\frac{\epsilon}{p},\cdots,i_\ell+\frac{\epsilon}{p}\}\subseteq \tI$;%
\footnote{This implies that the blocks $B_{i_1},\dots,B_{i_\ell}$ in $x$ are disjoint.}
(ii) the blocks $\bar{B}_{\bar{}i_1},\dots,\bar{B}_{\bar{i}_\ell}$ in $y$ are disjoint;
(iii) $A^p(i_j)\neq \bar{i}_j$ 
or $A^p(i_j+\frac{2\epsilon}{p}) \neq \bar{i}_j+\frac{2\epsilon}{p}+\gamma_j$
(possibly both);
and
(iv) $\ED(B_{i_j},\bar{B}_{i_j})=1$. 
Notice that conditions (i), (ii) and (iii) actually depend only on $A^p$,
and thus can be viewed as restrictions on the choice of $S$, $\bar{S}$ and $\gamma$ in~\eqref{eq:alignbound1};
with this viewpoint in mind, we can simply write 
\[
  \mE_{S,\bar{S},\gamma} 
  \eqdef 
  \set{ \ED(B_{i_1},\bar{B}_{i_1})= 1,\ldots, \ED(B_{i_\ell},\bar{B}_{i_\ell})= 1 }.
\]

We proceed to prove~\eqref{eq:alignbound1}. 
Suppose event $\mathcal{B}$ occurs, i.e., there exists an alignment $M$ between $x,y$ with $\cost(M)\le (1+\delta)pn$ such that $|\mathcal{\tilde{I}}_M|<(1-23\epsilon-\delta)pn$.
Next first assume the case where there exists at least $6\epsilon p n$ indices $i\in \mathcal{\tilde{I}}$ such that $\cost_M(x[i-\frac{\epsilon}{p},i+\frac{\epsilon}{p}])=0$. But then trivially event $\mathcal{\bar{B}}$ is satisfied and thus we prove~\eqref{eq:alignbound1}. Hence from now on wards we assume the case where event $\mathcal{\bar{B}}$
is not satisfied.                        
Note for each position $i\in\tI$, the intervals $[i-\frac{\epsilon}{p},i+\frac{\epsilon}{p}]$ in $x$ are disjoint (by definition of $\tI$), 
and thus by Lemma~\ref{lem:subcost}, 
\[
  \sum_{i\in\mathcal{\tilde{I}}}  
    \cost_M([i-\tfrac{\epsilon}{p},i+\tfrac{\epsilon}{p}]) 
  \le \cost(M)
  \le (1+\delta)p n.
\]

The number of summands here is $|\tI| \geq (1-5\epsilon)p n$. 
Moreover, as $\mathcal{\bar{B}}$ is not satisfied, 
the number of indices in $\mathcal{\tilde{I}}$ such that the associated block in $x$ has cost $0$ 
is at most $6\epsilon p n$, 
and thus at least $(1-11\epsilon) p n$ summands contribute cost at least $1$. 
Let $A$ be the set of indices $i$ that contribute cost $1$ to the above summation. 
Then 
\[
  \sum_{i\in\mathcal{\tilde{I}}}  
    \cost_M([i-\tfrac{\epsilon}{p},i+\tfrac{\epsilon}{p}]) 
  \ge |A|\cdot 1 + ((1-11 \epsilon) p n - |A|)\cdot 2
  = (2-22\epsilon) p n - |A|.
\]

Thus $|A|\ge (1-22\epsilon-\delta)p n$. Let $\tilde{S}=A\setminus \mathcal{\tilde{I}}_M$ and by our assumption $|\tilde{S}|\ge \epsilon pn$.
To get the exact size $|\tilde{S}|=\epsilon pn$, 
we can replace $\tilde{S}$ with an arbitrary subset of it of the exact size. 

Now define $S=\set{i-\frac{\epsilon}{p}\mid i\in \tilde{S}}$, $\bar S = \set{\min_{k\in\{0,1\}} M(i+k) \mid i\in S; M(i+k)\neq \bot}$.
For each $i\in S$ define $N_i=\max_{k\in[i,i+\frac{2\epsilon}{p}]; M(k)\neq \bot}M(k)-\min_{k\in[i,i+\frac{2\epsilon}{p}]; M(k)\neq \bot}M(k)$.
Define
$\gamma=\{N_i-\frac{2\epsilon}{p}\mid i\in S\}$. 
Let us verify that the event $\mE_{S,\bar{S},\gamma}$ holds. 
Indeed, $\tilde{S}\subset \mathcal{\tilde{I}}$. Moreover as for each $i\in S$, 
$\cost_M[i, i+\tfrac{2\epsilon}{p}] = 1$,
which implies $ N_i\in\{\frac{2\epsilon}{p}-1,\frac{2\epsilon}{p}+1\}$. 
Moreover, the block $x[i, i+\tfrac{2\epsilon}{p}]$ in $x$ 
is at distance at most 1 from the corresponding block in $y$, 
and these blocks in $y$ are disjoint. Also for each $i\in S$; $i+\frac{\epsilon}{p}\in \mathcal{\tilde{I}}\setminus \mathcal{\tilde{I}}_M$. Thus either $A^p(i)\neq \min_{k\in[i,i+\frac{2\epsilon}{p}]}\{M(k)\mid M(k)\neq \bot\}$. In this case $\bar{i}\neq A^p(i)$ and otherwise $A^p(i+\frac{2\epsilon}{p})\neq \max_{k\in[i,i+\frac{2\epsilon}{p}]}\{M(k)\mid M(k)\neq \bot\}$ and thus $\bar{i}+\frac{2\epsilon}{p}+\gamma_j\neq A^p(i+\frac{2\epsilon}{p})$ 
This completes the proof of~\eqref{eq:alignbound1}. 

Next, we prove~\eqref{eq:alignbound2}. 
Fix $S,\bar{S}\subset [n]$ of the same size $\ell$ and $\gamma\in\{-1,1\}^\ell$.
Assume requirements (i), (ii) and (iii) hold (otherwise, the probability is 0). 
Consider now a given $j\in[\ell]$. 
Let $B_{i_j}=x[i_j,i_j+\frac{2\epsilon}{p}]$ and $\bar{B}_{i_j}=y[\bar{i}_j,\bar{i}_j+\frac{2\epsilon}{p}+\gamma_j]$ be the corresponding blocks in $x$ and $y$ 
and let $\gamma_j$ be the corresponding value from $\gamma$.

\textbf{Case $\gamma_j=-1$:} 
Notice in this case the requirement $\ED(B_{i_j},\bar{B}_{i_j})= 1$ implies 
that block $\bar{B}_{i_j}$ is obtained by one deletion from $B_{i_j}$. 
We further divide this into two subcases.
The first subcase is when that deletion occurs in the "middle" interval $[i_j+\frac{\epsilon}{2p}, i_j+\frac{3\epsilon}{2p}]$ (including the boundary points). 
This implies $\forall t\in \{0,\dots,\frac{2\epsilon}{p}\}\cup\{\frac{3\epsilon}{2p}+1,\dots,\frac{2\epsilon}{p}\}$, (a) the prefix of $B_{i_j}$ (i.e., $x[i_j,i_j+\frac{\epsilon}{2p}-1]$) satisfies $x[i_j+t]=y[\bar{i}_j+t]$ as the deletion occurs after this prefix and (b) the suffix of $B_{i_j}$ (i.e., $x[i_j+\frac{3\epsilon}{2p}+1,i_j+\frac{2\epsilon}{p}]$) satisfies $x[i_j+t]=y[\bar{i}_j+t-1]$ as the deletion occurs before this suffix.
Now again $x$ and $y$ are random but correlated through $A^p$; 
thus $x[i_j+t]$ and $y[\bar{i}_j+t]$ (or $y[\bar{i}_j+t-1]$) are chosen independently at random unless $A^p$ aligns their positions, i.e., $A^p(i_j+t)=\bar{i}_j+t$ (or $=\bar{i}_j+t-1$). 
This cannot happen for both $t=0$ and $t=\frac{2\epsilon}{p}$ 
as we assumed that condition (iii) is satisfied. 
If $A^p(i_j+t)\neq \bar{i}_j+t$ for $t=0$,
then the same holds for all $t=1,\cdots, \frac{\epsilon}{2p}-1$,
because $A^p$ has no edit operation in the interval $[i_j,i_j+\frac{\epsilon}{2p}-1]$,
and thus
$$A^p(i_j+t) = A^p(i_j)+t \neq \bar{i}_j+t.$$
If $A^p(i_j+t)\neq \bar{i}_j+t-1$ for $t=\frac{2\epsilon}{p}$,
then by a similar argument, the same inequality holds for all  $t=\frac{3\epsilon}{p}+1,\dots,\frac{2\epsilon}{p}-1$.
Either way, we conclude that for the event $\ED(B_{i_j},\bar{B}_{i_j})= 1$ to occur, 
at least $\frac{\epsilon}{2p}$ constraints of the form $x[i_j+t]=y[\bar{i}_j+t]$ must be satisfied, 
where these two positions are not aligned by $A^p$, 
and thus these two symbols are chosen independently at random. 

Next, consider the subcase of a deletion outside the "middle" interval of $B_{i_j}$,
i.e., in $[i_j,i_j+\frac{\epsilon}{2p}-1]\cup [i_j+\frac{3\epsilon}{2p}+1,i_j+\frac{2\epsilon}{p}]$. 
Let us first assume that the deletion occurs from interval $[i_j,i_j+\frac{\epsilon}{2p}-1]$.
This implies for all $t\in[\frac{\epsilon}{2p},\frac{3\epsilon}{2p}]$, 
$x[i_j+t]=y[\bar{i}_j+t-1]$. 
Again we can claim the two symbols $x[i_j+t]$ and $y[\bar{i}_j+t-1]$ are chosen independently at random unless $A^p$ aligns their positions, $A^p(i_j+t)=\bar{i}_j+t-1$. Now observe that this cannot happen for both $t=\frac{\epsilon}{p}-1$ and $t=\frac{\epsilon}{p}+1$, because in that case, $A^p(i_j+\frac{\epsilon}{p}+1)-A^p(i_j+\frac{\epsilon}{p}-1)=\frac{\epsilon}{p}-(\frac{\epsilon}{p}-2)=2$; however $i_j+\frac{\epsilon}{p}\in \mathcal{\tilde{I}}$ implies that $A^p$ has exactly one edit operation (insertion or deletion) in the interval $[i_j+\frac{\epsilon}{p}-1,i_j+\frac{\epsilon}{p}+1]$ (and not in the boundary points), thus $A^p(i_j+\frac{\epsilon}{p}+1)-A^p(i_j+\frac{\epsilon}{p}-1)\in \{1,3\}$. Assume first $A^p(i_j+t)\neq \bar{i}_j+t-1$ for $t=\frac{\epsilon}{p}+1$. 
Then the same must hold also for all $t=\frac{\epsilon}{p}+2,\cdots,\frac{3\epsilon}{2p}$, 
because $i_j+\frac{\epsilon}{p}\in \mathcal{\tilde{I}}$ 
and this implies that $A^p$ has no edit operation in the interval $[i_j+\frac{\epsilon}{p}+1,i_j+\frac{3\epsilon}{2p}]$, 
thus 
$$A^p(i_j+t)=A^p(i_j+\tfrac{\epsilon}{p}+1)+(t-\tfrac{\epsilon}{p}-1)\neq (\bar{i}_j+\tfrac{\epsilon}{p})+(t-\tfrac{\epsilon}{p}-1) = \bar{i}_j+t-1. $$
We can apply a similar argument for $t=\frac{\epsilon}{p}-1$, 
and again we conclude that, 
for the event $\ED(B_{i_j},\bar{B}_{i_j})= 1$ to occur,
at least $\frac{\epsilon}{2p}$ constraints of the form $x[i_j+t]=y[\bar{i}_j+t-1]$ must be satisfied where these two positions are not aligned by $A^p$, 
and thus these two symbols are chosen independently at random. For the case where the deletion occurs in the interval $[i_j+\frac{3\epsilon}{2p}+1,i_j+\frac{2\epsilon}{p}]$, following a similar argument we can show if $\ED(B_{i_j},\bar{B}_{i_j})= 1$ then for each $t\in[\frac{\epsilon}{2p},\frac{3\epsilon}{2p}]$, 
$x[i_j+t]=y[\bar{i}_j+t]$ where again these two positions are not aligned by $A^p$ and hence the two symbols are chosen independently at random.

\textbf{Case $\gamma_j=1$:}
We can show that if $\ED(B_{i_j},\bar{B}_{i_j})= 1$ then at least $\frac{\epsilon}{2p}$ requirements of the form $x[i_j+t]=y[\bar{i}_j+t]$ (or $=y[\bar{i}_j+t+1]$) must be satisfied where these two positions are not aligned by $A^p$, and thus these two symbols are chosen independently at random. 

The above argument applies to every $j\in[\ell]$,
yielding overall at least $\ell\cdot \tfrac{\epsilon}{2p}$ requirements of the form  
$x[i_j + t] = y[\bar{i}_j+t']$ (where $t'\in\{t-1,t,t+1\}$),
where these two symbols are chosen independently at random. 
Observe that each $y[\bar{i}_j+t']$ is 
either a character $x[k]$ (for $k$ arising from $A^p$) or completely independent. 
Since each character of $x$ appears in at most $2$ requirements (once on each side),
we can extract a subset of at one-third of the requirements
such that the positions in $x$ appearing there are all distinct,
and thus the events are independent.%
We overall obtain at least $\tfrac13 \ell\cdot \tfrac{\epsilon}{2p}$ requirements,
each occurring independently with probability $1/|\Sigma|$,
and thus
\[
  \Pr[ \mE_{S,\bar{S},\gamma} ] 
  \leq 
  |\Sigma|^{ -\epsilon\ell/(6p) } .
\]

Finally, we are in position to prove the claim. 
By Lemma~\ref{lem:zerocostbound}, 
$\Pr[\mathcal{B}^0]\le 2e^{-\epsilon^2 p^2 n/2}$.
Combining \eqref{eq:editbound1} and \eqref{eq:editbound2} and a union bound
\begin{align*}
  \Pr[ \mathcal{B} ]
  &\leq \binom{n}{\ell}^2 \cdot 2^\ell \cdot |\Sigma|^{ -\epsilon\ell/(6p) } +2e^{-\epsilon^2 p^2 n/2}
  \\
  &\leq \Big( \frac{n e}{\ell} \Big)^{2\ell}\cdot 2^\ell \cdot 2^{ -\epsilon\ell/(6p) } +2e^{-\epsilon^2 p^2 n/2}
  \\
  &\leq \Big( \frac{e}{\epsilon p } \Big)^{2\epsilon pn} \cdot 2^{\epsilon pn}\cdot 2^{-\epsilon^2 n/6 } +2e^{-\epsilon^2 p^2 n/2}\\
  &\leq (p^2)^{-2\epsilon pn} \cdot 2^{2\epsilon pn}\cdot 2^{ -\epsilon^2n/6}+2e^{-\epsilon^2 p^2 n/2}
  \\
  &\leq (p)^{-6\epsilon pn} \cdot 2^{ -\epsilon (42p\log(1/p)) n/6 } +2e^{-\epsilon^2 p^2 n/2}
  \\
  &\leq p^{-6\epsilon pn +7\epsilon pn}+2e^{-\epsilon^2 p^2 n/2}
  \\
  &\leq p^{ \epsilon pn}+2e^{-\epsilon^2 p^2 n/2}.
\end{align*}
Recall that this was all conditioned on $A^p$,
which had error probability at most $e^{-{\epsilon^2 p^2 n}/{2}}$,
and now Lemma~\ref{lem:unique-alignment} follows by a union bound. 
\end{proof}

\end{document}